 \documentclass[reqno,11pt,a4paper]{amsart}

\usepackage{amsfonts,amsmath,amssymb}
\usepackage[latin1]{inputenc}
\usepackage{dsfont}
\usepackage{enumerate}
\usepackage{xcolor}
\usepackage[all]{xy}
\def\notshow#1\notshowend{} %
\usepackage{hyphenat}

\usepackage{graphicx}

\usepackage{tikz}
\usepackage{tikz-3dplot}
\usepackage{pgfplots}

\usepackage{multirow}
\usepackage{array}

\newcommand{\C}{\mathcal{C}}

\usepackage{amssymb}
\usepackage{amsfonts}
\usepackage{mathrsfs}
\usepackage{hyperref}

\usepackage[normalem]{ulem} 

\def\bb#1\eb{\textcolor{blue}{#1}} 
\def\br#1\er{\textcolor{red}{#1}} %
\def\bm#1\em{\textcolor{magenta}{#1}} %

\newcommand{\R}{\mathds R}

\newcommand{\n}{\mathrm{n}}

\newcommand{\ok}{O^{(k)}(4,\R)}
\newcommand{\ah}{\mathbf{a_h}}
\newcommand{\tah}{\mathbf{{\tilde a}_h}}
\newcommand{\av}{\mathbf{a_v}}
\newcommand{\tav}{\mathbf{\tilde{a}_v}}

\newcommand{\ttt}{\mathbf{t}}

\newtheorem{thm}{Theorem}[section]
\newtheorem{prop}[thm]{Proposition}
\newtheorem{lemma}[thm]{Lemma}

\theoremstyle{definition}
\newtheorem{defi}[thm]{Definition}

\newtheorem{exe}[thm]{Example}

\newtheorem{rem}[thm]{Remark}
\newtheorem{note}[thm]{Note}

\newtheorem{postulate}{Postulate}
\newtheorem{exer}{Exercise}

\newcommand{\ben}{\begin{enumerate}}
\newcommand{\een}{\end{enumerate}}
\newcommand{\bit}{\begin{itemize}}
\newcommand{\eit}{\end{itemize}}
\newcommand{\edoc}{\end{document}}

\newcommand{\Aff}{\hbox{Aff}}

\newcommand{\Si}{\mathbf{S}_{\hbox{{\tiny IFR}}}}
\newcommand{\So}{\mathbf{O}}
\newcommand{\Sm}{\mathbf{S}}

\title[Foundations of Finsler spacetimes]{FOUNDATIONS OF FINSLER SPACETIMES
FROM THE OBSERVERS' VIEWPOINT}

\author[A. Bernal]{Antonio N. Bernal} \address{Departamento de Geometr\'{\i}a y Topolog\'{\i}a, Facultad de Ciencias, \hfill\break\indent Universidad de Granada,\hfill\break\indent Campus Fuentenueva s/n, \hfill\break\indent 18071 Granada, Spain}\email{anbernal@no-gravity.eu}

\author[M. A. Javaloyes]{Miguel Angel Javaloyes}\address{Departamento de Matem\'aticas, \hfill\break\indent Universidad de Murcia, \hfill\break\indent Campus de Espinardo,\hfill\break\indent 30100 Espinardo, Murcia, Spain} \email{majava@um.es}

\author[M. S\'anchez]{Miguel S\'anchez} \address{Departamento de Geometr\'{\i}a y Topolog\'{\i}a, Facultad de Ciencias, \hfill\break\indent Universidad de Granada,\hfill\break\indent Campus Fuentenueva s/n, \hfill\break\indent 18071 Granada, Spain}\email{sanchezm@ugr.es}

\begin{document}
\begin{abstract} Physical foundations for relativistic spacetimes are revisited, in order to check at what extent Finsler spacetimes lie in their framework. Arguments based on inertial observers (as in the foundations of  Special Relativity and  Classical Mechanics)  are shown to correspond with a double linear approximation in the measurement of space and time. While General Relativity 
appears by dropping the first linearization,   Finsler spacetimes appear by dropping the second one. The classical Ehlers-Pirani-Schild approach is  carefully  discussed and shown to be compatible with the Lorentz-Finsler case.
 The precise 
 mathematical definition of Finsler spacetime is discussed by using the {\em space of observers}. Special care is taken in some issues such as: 
the fact that 
a Lorentz-Finsler metric would be physically 
 measurable only on the causal directions for a cone structure,  
the implications for models of spacetimes of  some apparently innocuous hypotheses on differentiability, or 
the possibilities of measurement of a varying speed of light. 

\vspace{5mm}

\end{abstract}
\maketitle


\tableofcontents

\section{Introduction}
A plethora of alternatives to classical General Relativity has been developed since its very beginning. Many of them were motivated by the search of a unified theory which solved disturbing issues of compatibility with Quantum Mechanics 
 (Kaluza-Klein, M-theory, quantum field gravity...) while, since the 90's,  unexpected cosmological measurements led to further alternatives (cosmological constant, quintaessence,  theories with varying speed of light...). However, the possibility to consider a Finslerian modification of GR  has not settled   in the mainstream of research and it has been scarcely considered in the literature until recent times (some examples are references \cite{AaJa16,CaStan16,CaStan18,FP16,
 FPP18,GGP07,HPV,Ishi81,JavSan14,
 JavSan18,LP,Min16,PerlickFermat,
 Pfeifer,Tavakol}). Certainly, the generality of Finsler Geometry in comparison with  the Riemannian setup  (namely, analogous to the generality of the convex open subsets of an affine space in comparison with the ellipsoids) is a big drawback, as the number of new variables and parameters would  seem 	immeasurable.  Neverthelesss, this is similar to the generality of General Relativity in comparison with Special one (see Remark \ref{r_6.1}).  Anyway,  any Finslerian modification of General Relativity would mean to drop  the beloved Lorentz invariance not only at  global and local levels  (as it occurs  in  General Relativity)  but also infinitesimally, i.e. looking such an invariance as a limit symmetry around each event. However, from a fundamental viewpoint, this should  not seem too strange: as physical measurements are always approximations, one would not be surprised if the  symmetries of the models were only approximations to a more complex reality. Indeed, as we will explain, the existence  of some symmetries among observers becomes a   natural requirement   in order to make direct measurements of space and time. There is no  reason to assume that the physical reality will satisfy such requirements in an exact way ---even though, certainly, the existence of such approximated symmetries are meaningful and useful for modeling. 
 
In the present article, a physical motivation to consider Finsler spacetimes as models of  space and time is developed, and quite a few of related ideas are discussed. We stress the following four guidelines.

\subsubsection*{1. Approach from the foundations viewpoint} 
We develop an approach for the foundations of the theories of spacetime 
starting at the observers viewpoint in  Classical Mechanics and Special Relativity (\S \ref{s2}--\ref{s4}).
 Finsler spacetimes are shown to appear by dropping the symmetries of inertial observers in a natural way.
 Our approach 
follows the viewpoint in  \cite{BLS} by  L\'opez and two of the authors in \S \ref{s2} and \S \ref{s3}, 
 which includes  the celebrated ideas by V. Ignatowski  \cite{Ignato, Ignato1, Ignato2} about  the foundations of Special Relativity.
 
Specifically, we argue that the geometric models of spacetime appear from the notion of inertial observers by means of a double linearization of the measuring problem, namely: \bit\item[(1)] there are inertial frames of reference (IFR) whose changes of coordinates  are linear, and 
 \item[(2)] the symmetries in the change of the timelike coordinate (and, independently, in the three spacelike ones) between two IFR's are encoded in that linear structure. \eit 
 These assumptions lead to four classic linear 4-dimensional structures, name\-ly Lorentz-Minkowski, Galilei-Newton, dual Galilei-Newton and Euclidean (see \S \ref{s2}).  
 Dropping (1)  leads from Special to General Relativity (see~\S \ref{s3}), as well as to other transitions  for the other three structures. The latter are  also briefly explained here, namely, from Galilean to Leibnizian spacetimes, see \S \ref{s3.4}, including signature changing metrics (see \S \ref{s3.1}),  and  the possibility of a pointwise varying speed of light  (see \S \ref{s3varying_c}).  Dropping  (2) leads to Lorentz-Minkowski norms and, then, to Finsler spacetimes, discussed both mathematically and physically in \S  \ref{s4}.

\subsubsection*{2. Critical revision of EPS axiomatics}
The classical Ehlers, Pirani and Schild (EPS) approach  for General Relativity \cite{EPS} will be revisited (see \S \ref{s5}). We show that, certainly, this approach {\em is compatible} with the existence of Lorentz-Finsler metrics,  a  possibility  already suggested  by  Tavakol and Van Den Bergh for Berwald spaces \cite{Tavakol}  (see the discussion in \S \ref{s_5.2.5}).  Such   a possibility was ignored  in EPS  because of  a too restrictive development of two steps\footnote{The reason relies on a classical result for  any Finsler metric $F$: its square $F^2$  is $C^2$ at the zero section if and only if $F$ comes from a Riemannian metric (see Remark \ref{r_finslerm} (1) and \S \ref{5.2}).}, 
namely:

\ben\item   An artificial  requirement of smoothability  of some combination of radar coordinates,  which  would forbid null cone structures non-compatible with Lorentzian metrics.  This was recently pointed out by  Lamm\"erzhal and Perlick's \cite{LP} and it is developed here in detail   (see \S \ref{5.2.1}). 
\item  A deduction of the existence of a projective structure  starting at a general version of the law of inertia. 
 This would exclude the timelike pregeodesics for a Lorentz-Finsler
metric (except those of Berwald-type)
but, again, the proof crucially relies on an argument of $C^2$-differentiability, which is related to
non-trivial issues on Finslerian metrics 
 (see \S \ref{5.2.2}). 
\een

\subsubsection*{3. Precise geometric framework}
Along the article, a careful mathematical approach is carried out, following \cite{JavSan18}. For the convenience of the reader,   a brief mathematical summary on (Lorentz) Finsler  concepts is also included,  \S \ref{s4.1}.  This allows us to model and to discuss issues on Lorentz-Finsler metrics $L$ which turn out to be important from the physical viewpoint such as:

\ben\item {\em Causal cone domain} (\S \ref{s4}). The  physically meaningful domain for $L$ is only  the causal cone of a cone structure $\C$. 

 Indeed, even in the classical relativistic case  only the future-directed causal directions for a cone $\C^+$ determined by the metric $g$ contains the elements physically measurable for any (true or {\em gedanken}) experiment. 
 In Relativity,   the  Lorentzian scalar product $g_p$ at each event $p$ is    determined by its value on the   cone $\C^+_p$ (or on its timelike directions); therefore, a Lorentz metric $g$  can be determined on the whole $TM$ even if, actually, only its value on $\C^+$ can be measured.  
However, this is not by any means true for a Lorentz-Finsler metric $L$, where there is a huge freedom to extend the Lorentz-Finsler metric away from $\C$. 

So, our Lorentz-Finsler metrics will be defined only on a (causal) cone 
structure\footnote{This is consistent with our choices in our previous work \cite{JavSan14}. There are other reasons for this choice from the purely  mathematical  viewpoint, as it clarifies the properties of anisotropically conformal metrics, see \cite{JavSan18}.}.  

\item {\em Smoothness}, i.e., differentiability up to some appropriate order. Usually, such a requirement is regarded as a harmless macroscopic approximation to the structure of the spacetime. However, the discussion on EPS above shows that this is not so trivial in the Finslerian case.  What is more,  other issues  appear in the literature:
\bit\item The possibility that  the cone is smooth and the Lorentz-Finsler metric is smooth only on the timelike directions  but cannot be smoothly extended to the cone, which happens in metrics such as Bogoslovski in Very Special Relativity \cite{Bogoslovsky} and others \cite{PW11},  see  \S \ref{s_6.1_modifiedSpecialRelativity}. 
\item The lack of differentiability  outside the zero section of Finsler product spacetimes, which may lead to definitions of Finsler     static spacetimes which are not smooth in the static direction \cite{CaStan16}, a fact which can be overcome with our approach  to the space of observers,  see \S \ref{s_phys_intuitions} (item \ref{item5} (b)).

\eit
\item {\em Anisotropic speed of light.} Finsler spacetimes permit different possibilities for a speed of light which may vary  not only with the point (an issue already considered even for relativistic spacetimes, \S \ref{s3varying_c})  but also with the direction, \S \ref{s6.3}. 
\een

\subsubsection*{4. Importance of the space of observers}  The relevance of the {\em space of observers} in Special and General Relativity, its links with the 
symmetries of the spacetime and  the possibility to lift Relativity to this space have been stressed by several authors \cite{GW,Hohmann} in the framework of Lorentz violation and Lorentz-Finsler geometry. 
 It is worth emphasizing that the  essential role of this space appears explicitly  along our development.  In the initial discussion of the linearized models, we start with the set $\Sm$ of  {\em inertial frames of reference} (IFR), which permits even signature changing metrics (\S \ref{s3.1}). However, once the symmetries of these models are dropped, only the space of observers $\So$ remains as physically meaningful (Definition~\ref{d_So}). 
In a classical relativistic spacetime $(M,g)$, $\So$ is just the submanifold $\Sigma^g\subset TM$ of all the $g$-unit vectors in the future timelike  cone; thus,  each $\Sigma^g_p:=\Sigma^g\cap T_pM$ is a hyperbolic space in the tangent space $T_pM$ of each event $p\in M$. Breaking  Lorentz symmetry at each $p$ leads to regard $\Sigma^g$ just as a  more general pointwise  concave hypersurface $\Sigma$, which becomes then the indicatrix of a Lorentz-Finsler metric $L$  (see Remark \ref{rem_items}). 

This observers' viewpoint allows one to use geometric methods recently developed in \cite{JavSan18} which may have interesting physical applications such as: (a) going from $g$ to $L$ by perturbing   the pointwise hyperboloids $\Sigma^g$ into pointwise  concave hypersurfaces  $\Sigma$ (as suggested in \S \ref{s6.4}, such a perturbation might be produced by the presence of matter/energy and lead to quantum consequences), (b) to avoid or to smoothen possible singularities in $\Sigma$ and then in $L$ (showing that known non-smooth physical examples can always be approximated by smooth ones), (c) to construct systematically any Lorentz-Finsler metric from a Riemannian and a Finslerian one or (d) to construct systematically  static and stationary metrics (avoiding any problem of smoothability). 

 In our opinion, the previous ideas support strongly that  Finsler spacetimes have become an exciting vast field to explore thoughtfully from both the physical and mathematical viewpoints.

\section{The doubly linearized models}\label{s2}

 Next, we develop our approach for the foundations of the theories of spacetime.
As a difference with the EPS approach, we will not assume postulates on {\em the nature of the behavior of the physical objects} which will be measured but on {\em how we can measure} those physical objects. 
A posteriori, if we are able to measure by using some sort of symmetry, the spacetime itself will be endowed with the geometric structure which codifies such symmetries. 

 The first step, to be developed along this section, considers the simplest symmetries for observers, common to both Classical Mechanics and Special Relativity.  They will be regarded later as a (linear) idealization. 

\subsection{Postulates}  Let us introduce the approach to the 
theories of spacetimes following\footnote{It is worth pointing out that \cite{BLS} focuses on the viewpoint of General Relativity. So, the first postulate there is different to the one here.   Our  viewpoint was pointed in the  reference \cite{BS_gaceta} (written for a general audience  in Spanish) and it is developed further here by introducing concepts such as apparent temporality (Theorem \ref{t_k}) or arguments as those on the varying of the speed of light. } 
\cite{BLS} (a priori, this is non-quantum, even though quantum links  will appear in \S \ref{s6.4}). 

The physical considerations on the existence of inertial  
frames of reference are encoded in the following two postulates. 

\begin{postulate}[Linear approach to  spacetime] \label{p1} The physical spacetime is endowed with a structure of affine space $\Aff$ on a  real  vector space $V$ of dimension $n=4$. 
Physical observers are able to construct a non-empty set $\Si$ of affine frames of reference (each one $R=(O,B)$ composed by a point $O\in \Aff$ and a  basis  $B$ of $V$) which are called {\em inertial frames of reference} (IFR). 

Thus, each IFR, $R$, provides an affine chart, i.e. a bijection  $\varphi: \Aff\rightarrow \R^4$, $\varphi(P)=(t(P),x^1(P),x^2(P),x^3(P))$, such that, given another IFR, $\bar R$, the coordinate change  $\bar\varphi\circ \varphi^{-1}: \R^4\rightarrow \R^4$ is an affine map. The first coordinate $t$ of each IFR will be called {\em temporal} and the other three $x^i$, {\em spatial}. 
\end{postulate}
The meaning of this first postulate is just 
that a linear approximation  Aff  to spacetime  is being considered. The postulate also says that physicists will be able to construct {\em some} of the natural charts of the affine space $\Aff$. The physical process to obtain such charts  is not specified, even though the names {\em temporal} and {\em spatial} suggest the nature of their measurements. 

Our second postulate, based essentially in von Ignatowski's \cite{Ignato}, will ensure just that, when making measurements of the temporal coordinate (resp. when making measurements of the spatial coordinates), the viewpoint of two IFR's are interchangeable. 
This will be reflected by a requirement of  symmetry in the corresponding charts. To understand this symmetry easily, let us discuss the bidimensional case $n=2$.  
Let $R, \bar R$ be two IFR's with coordinates $(t,x)$ and $(\bar t,\bar x)$, resp. By Postulate~1:
\begin{equation}
\label{ep1}
\begin{pmatrix}
\bar t\\
\bar x
\end{pmatrix}
=\begin{pmatrix}
a & b\\
c & d 
\end{pmatrix}
\begin{pmatrix}
t\\
x
\end{pmatrix}
+\begin{pmatrix}
e\\
f
\end{pmatrix}.
\end{equation}
The interchangeability of the viewpoints of $R$ and $\bar R$ will collect the following physical assertion:  the temporal coordinate $\bar t$  (resp. the spatial coordinate $\bar x$) of $\bar R$  measured by using the physical clock (resp. the rod) of $R$  goes by as  the temporal coordinate $t$  (resp. the spatial coordinate $x$) of $R$  measured by using the physical clock (resp. the rod) of $\bar R$. Mathematically,
\begin{align}
\label{ep2}
\partial\bar t/\partial t\,(=a)&=\partial t/\partial\bar t
&&\mbox{and}
&\partial\bar x/\partial x\,(=d)&=\partial x/\partial\bar x .
\end{align}
In dimension $n=4$, interchangeability between the three spatial coordinates will also be imposed. 

\begin{postulate}[Time and spatial interchangeability] \label{p2} Let $R, \bar R \in \Si$  be  two IFR's. Then, their coordinates $(t, x_1, x_2, x_3)$ and $(\bar t, \bar x_1, \bar x_2, \bar x_3)$ satisfy:
\begin{align}\label{e_post2}
\partial\bar t/\partial t&=   \partial t/\partial\bar t,
&&
&\partial \bar x_i/\partial x_j&=\partial x_j/\partial\bar x_i, \qquad \qquad \forall i,j=1,2,3.
\end{align}
\end{postulate}

\subsection{Groups $O^{(k)}(4,\mathbb{R})$} The linear part of an affine change of coordinates from a first IFR, $R$, to a second one, $\bar R$,  will be called the {\em transition matrix} $A$ from $R$ to $\bar R$.  The second postulate implies that the transition matrices satisfy  the condition \eqref{e_exercise1} below,  so,   in order to  obtain all the possibilities, one just  needs to solve the following algebraic exercise.

\begin{exer}\label{ex_1} Let $A\in GL(4,\R)$ be a regular $4\times 4$ matrix and $A^{-1}$  its inverse. Write them by using boxes as follows:
\begin{align*}
A= &
\left(
\begin{array}{c|c}
a_{00} & \ah\\
\hline
\av^t & \hat A
\end{array} \right), 
&
A^{-1}&=\left(
\begin{array}{c|c}
\tilde a_{00} & \tah\\
\hline
\tav^t & \tilde A
\end{array} \right), 
\end{align*}
where $a_ {00},\tilde a_{00}\in\R$, 
$\ah,\av, 
\tah, 
\tav\in \R^3$, 
$\hat A$, 
$\tilde A$  are $3\times 3$ submatrices, and the superscript $^t$ denotes transponse. Then, determine those matrices $A$ satisfying:
\begin{equation}\label{e_exercise1}
\tilde a_{00}= a_{00}\qquad \qquad \tilde A= \hat A^t.
\end{equation}
\end{exer}
Such an exercise is solved in \cite[\S 3]{BLS} in full detail. Next, we will describe the main properties of its solutions\footnote{The reader can consider the simple case $n=2$ (as  in \eqref{ep2}), when $\hat A \equiv d$, $ \tilde A  \equiv \tilde d$ ($d,\tilde d\in \R$). 
	The  solutions of this case yield all the relevant possibilities. They follow easily by noticing that, from the algorithm to compute the inverse matrix: \begin{align*}
a&=\tilde a=d/\det A, & d&=\tilde d=a/\det A, & \tilde b&=-b/\det A & &\mbox{and} & \tilde c&=-c/\det A.  
\end{align*} In particular, $d\not=0$ $\Leftrightarrow$ $a\not=0$  and, then,  $\det A^2=1$. Therefore, this equality  would follow by  assuming additionally $a>0$  (i.e., $\partial \tilde t/\partial t>0$ in \eqref{e_post2}),  which will correspond with the condition of apparent temporality in Theorem \ref{t_k}.}. 

\begin{defi} Let  $S^1=\R\cup \{ \omega\}$ be the  circle regarded as the extended real line 
 $\mathbb{R}^*= [-\infty,+\infty]$   with $+\infty$ identified to $-\infty$ as a single point  $\omega$. For each $k\in S^1$, consider the matrix
 $$
I^{(k)}= \left(
\begin{array}{c|c}
k & 0\\
\hline
0 &  I_3
\end{array} \right) \qquad \hbox{(where} \; I_3 \; \hbox{is the} \; 3\times 3 \; \hbox{identity matrix)}
 $$ 
  and define  the group $O^{(k)}(4,\mathbb{R})\subset GL(4,\R)$ as follows:
\bit\item if $k\in \R$,  
$O^{(k)}(4,\mathbb{R}) = \{ A\in GL(4,\mathbb{R}):
\hbox{det}A^2=1, \; A^t I^{(k)}A=I^{(k)}\}$,
\item if $k=\omega$,  
$O^{(\omega)}(4,\mathbb{R}) = \{ A\in GL(4,\mathbb{R}): A^t\in O^{(0)}(4,\R)\}$.
\eit
We will say that $A\in GL(4,\R)$ is  {\em $k$-congruent} if $A\in O^{(k)}(4,\R)$.  Accordingly, two IFR's $R, R'$ are {\em $k$-congruent} is so 
is its transition matrix.  It is easy to check that any $k$-congruent matrix $A$ is  a solution of Exercise \ref{ex_1}  as in this case $A^t I^{(k)}= I^{(k)} A^{-1}$.  Remarkably, it will turn out that  the converse holds except in very exceptional cases (detailed in \cite[Prop. 3.1]{BLS}).  Indeed, these exceptional cases will be avoided by using very mild and natural conditions from both  the mathematical and physical viewpoints (any of the hypotheses (1)--(4) in the main Theorem \ref{t_k} below). 
\end{defi}
\begin{rem}\label{rem_2.2} (1) In the case $k\neq 0,\omega$, the equality 
$$
A^t I^{(k)}A=I^{(k)}
$$
implies $\hbox{det}A^2=1$ trivially. What is more, this equality is equivalent to
$$A^{-1} I^{(1/k)}(A^{-1})^t=I^{(1/k)}. $$
Then, the case $k=\omega$ becomes equivalent to taking the limit $k\rightarrow \omega  (\equiv \pm \infty)$: 
$$O^{(\omega)}(4,\mathbb{R}) = \{ A\in GL(4,\mathbb{R}):
\hbox{det}A^2=1, \; A^{-1} I^{(0)}(A^{-1})^t=I^{(0)}\}.$$

(2) If $A$ is $k$-congruent for two distinct values of $k$, then so it is for any $k$. Concretely, let $k_1,k_2\in S^1$, from  \cite[Lemma 3.3]{BLS} (see its part 1 and  proof): 
$$k_1\neq k_2 \Longrightarrow O^{(k_1)}(4,\R)\cap O^{(k_2)}(4,\R) = \cap_{k\in S^1} O^{(k)}(4,\R)= 
\{\pm 1\} \times O(3,\R),$$
where $O(3,\R)$ is the usual orthogonal group and $$\{\pm 1\} \times O(3,\R):= \left(
\begin{array}{c|c}
\pm 1 & 0\\
\hline
0 &  O(3,\R)
\end{array} \right).$$
\end{rem}
Now, the relevant solutions   
to  our exercise can be easily  described. 
\begin{lemma}\label{l_2.3}
Let $A\in GL(4,\R$) satisfy the hypothesis \eqref{e_exercise1} of Exercise \ref{ex_1}. 
\ben
\item If the matrix $A^2$ also satisfies the property \eqref{e_exercise1} then $\det A^2=1$.
\item If $\det A^2=1$, then there exists  $k\in S^1$ such that $A$ is $k$-congruent, and  either $k$ is unique or it can be   arbitrarily chosen in $S^1$.
\item Let $A_1, A_2 \in GL(4,\R)$ be
 $k_1$- and $k_2$-congruent,  resp. If $k_1$ is univocally determined and  $A_1 \cdot A_2$ (resp.  $A_2 \cdot A_1$) is $k$-congruent for some $k \in S^1$, then   $A_1 \cdot A_2$   (resp. $A_2 \cdot A_1$) is $k_1$-congruent.
\een
\end{lemma}
\begin{proof} Assertion (1) follows from the sentence above \cite[Lemma 3.3]{BLS} (recall that, as explained at the beginning of the paragraph containing that sentence, {\em incongruent} means det$A^2  \neq 1$). 
	For (2), the existence of $k$  follows also from the paragraph above \cite[Lemma 3.3]{BLS} and the uniqueness from part 1 of \cite[Lemma 3.3]{BLS} regarding $S_p$ as a set of two congruent observers and $A$ as the transition matrix between them 
	or from Remark 2.2(2). Assertion (3) follows from part 1 of \cite[Lemma 3.3]{BLS} regarding $S_p$ as a set of three congruent observers with transition matrices $A_1$, $A_2$ and, say, $A_1\cdot A_2$ (and its inverses). Then, all of them must be $k'$-congruent for some $k'$ and, as $k_1$ was univocally determined, $k'=k_1$.
\end{proof}
Lemma \ref{l_2.3} implies that, under minimal realistic hypotheses,  any set  $\Si$  of IFR determines   (at least)  one  value of $k\in S^1$. Mathematically, such  realistic properties just ensure that  $\det A= \pm 1$, which would be related to the conservation of the volume.  Such a property might  also be postulated directly, nevertheless, there are other physically sound  weak hypotheses that  imply it. 

In order to formulate such hypotheses,  recall first that the set $\Si$ of IFR obtained from our postulates is rather arbitrary. For example,
 the unique restriction to its number of  elements comes from   
$\Si\neq \emptyset$; that is, one can  remove arbitrarily some  elements of $\Si$ (but not of all them) and this new set would satisfy  the postulates \ref{ep1} and \ref{ep2} too.  What is more, if there is some $k\in S^1$ such that $\Si$ is composed by (a small number of) $k$-congruent IFR's, one can enlarge $\Si$ by acting with the group $O^{(k)}(4,\R)$ obtaining a bigger set $S$ of compatible IFR's.  Notice that if there were a second $k'\neq k$ such that all IFR's in $\Si$ were $k'$-congruent, a different enlargement $S'$ could also be obtained.  These observations suggest the following construction. 
Given $\Si$, define \begin{equation}\label{e_sifrestrella}
\Si^*:= \cap_\alpha S_\alpha, 
\end{equation} where 
 each $S_{\alpha}$ is a set of affine reference frames satisfying: (i)~$S_{\alpha}$ includes  $\Si$, (ii)~the change of coordinates between any two elements of $S_{\alpha}$ satisfies the formula \eqref{e_post2} in Postulate \ref{ep2}, and (iii) $S_{\alpha}$  is {\em maximal} (i.e., not included in a bigger set satisfying the previous conditions (i) and (ii)). 
 
Recall: (a) $\Si^* (\supset \Si)$ is determined univocally by $\Si$, (b) physically, all the affine reference frames in  $\Si^*$ could be regarded as IFR's with the same    status   as those in $\Si$, and (c) mathematically, one would expect that the transition matrices between all the pairs of elements of $\Si^*$  had  a more natural  structure than $\Si$.

\begin{thm}\label{t_k}
Let $\Si$ be a set of IFR's (satisfying the Postulates~\ref{ep1} and~\ref{ep2}).
There exists $k\in S^1$ such that  the transition matrix $A\in GL(4,\R)$ of each transformation   of coordinates between two IFR's, $R_1$ and $R_2$, is  $k$-congruent for all $R_1, R_2 \in \Si$, whenever any of the following hypotheses hold:

\ben\item {\em Conservation of the IFR volume}: $\det A= \pm 1$, for any  transition matrix  $A$.
\item  {\em Transitivity:} if $A$ is the transition matrix from   a first IFR, $R_1\in \Si$, to a second IFR, $R_2  \in \Si$, then there exists an IFR, $R_0$, such that the transition matrix $A$ from $R_0$ to $R_1$ is equal to $A$. 
\item {\em Action by a group}:  the set of transition matrices $A$ between elements of  $\Si^*$ (as in \eqref{e_sifrestrella}) is a subgroup $G$ of $GL(4,\R)$.
\item {\em Apparent temporality}:  any transition matrix $A$ between elements of $\Si$ satisfies $a_{00}>0$ (with $a_{00}$  as in Exercise~\ref{ex_1};  recall also the discussion at  \S \ref{s2.4}). 
\een
Moreover, the existence of such a $k$ implies that the properties (1), (2) and (3) hold, being the group $G$ in (3) either $O^{(k)}(4,\R )$ or  the intersection of all of them,  i.e.,  $\{\pm 1\}\times O(3,\R)$. \end{thm}
\begin{proof}
First, let us prove the existence of the required $k$ under the hypothesis {\em (1)} and, then, let us check {\em (1)} $\Leftarrow$ {\em (2)} $\Leftarrow$ {\em (3)} and {\em (1)} $\Leftarrow$ {\em (4)}. Under {\em (1)}, the existence of some $k$ for each $A$ is ensured by part (2)
of Lemma \ref{l_2.3}. Then \cite[Lemma 3.3(1)]{BLS} (or part (3) of Lemma \ref{l_2.3}) 
ensures that one can choose the same  $k$ for all the transition matrices $A$ determined by pairs of elements in $\Si$.
 If the hypothesis {\em (2)} holds then $A^2$ is also a transition matrix between IFR, and part (1) of Lemma~\ref{l_2.3} implies that the  hypothesis {\em (1)} holds too. Analogously, {\em (3)} implies {\em (2)} trivially. 
Finally,  {\em (4)}  implies {\em (1)} from \cite[Lemma 3.1, item 1(ii)]{BLS}.

For the last assertion, let us check that, when such a $k$ exists, then {\em (3)} holds. Indeed, one of the sets $S_\alpha$ in the definition of  
$\Si^*$, name it $S_k$, can be chosen such that the group $O^{(k)}(4,\R)$ acts  transitively on $S_k$   (just choose $R\in \Si$ and take all the affine reference frames  $R'$ with transition matrix $A$ in $ O^{(k)}(4,\R)$).  So, when $k$ is univocally determined for one pair of  elements $R_1, R_2 \in \Si$, then $\Si^*=S_k$ and the hypothesis {\em (3)} holds with the group $G= O^{(k)}(4,\R)$. Otherwise, $k$ can be  arbitrarily  chosen by Lemma~\ref{l_2.3}(2), then
$\Si^*= \cap_{k\in S^1} S_k$ and {\em (3)} holds with the group $G=\{\pm 1\}\times O(3,\R)$ (see Remark~\ref{rem_2.2}(2)).
\end{proof}

\subsection{Linear models of spacetimes} \label{s2.3}
Theorem \ref{t_k} implies that, whenever one of its mild hypotheses {\em (1)---(4)} holds, the existence of a  set  $\Si$ of IFR's according to Postulates \ref{p1} and \ref{p2}, selects a group $G=O^{(k)}(4,\R)$ (or the intersection of all of them). As  the spacetime was represented by an affine space $\Aff$ on a vector space $V$ by postulate \ref{p1},  this vector space (and then $\Aff$) will be endowed automatically with the geometric structure invariant by  $G$. 
Let us study each case. 

\begin{enumerate}
\item Case $k\in(-\infty,0)$. By the definition of $\ok$, $V$ is  naturally endowed with a Lorentzian scalar product $\langle\cdot , \cdot \rangle_1$. 
Indeed, if $R=(O,B=(e_0,e_1,e_2,e_3))$ is any IFR, then the unique $\langle\cdot , \cdot \rangle_1$ such that $B$ is an orthornormal 
basis for it, up to the normalization of its first vector 
 \begin{equation}\label{e_norm}
 \sqrt{|\langle e_0 , e_0 \rangle_1|}= \sqrt{-k},
 \end{equation}
becomes independent of the chosen $R$. 
What is more, for $k=-1$, the group $\ok$ is the Lorentz group; otherwise, $\ok$ is conjugate to the Lorentz group. Indeed, putting $k=-c^2$ with $c>0$, then $I^{(k)}=I^{(c)}\cdot I^{(-1)}\cdot I^{(c)}$, 
the inverse of $I^{(c)}$ is  $I^{(1/c)}$ and 
\begin{equation}\label{e_conj}
\ok = I^{(1/c)}\cdot O^{(1)}(4,\R)\cdot I^{(c)}.
\end{equation} 
 Anyway, the spacetime of Special Relativity is obtained.
 
\item Case $k=\omega$. The group $\ok$ becomes the (non-orthochronous) {\em Galilean group}
$$
O^{(\omega)}(4,\mathbb{R}):=\left\{
\left(\begin{array}{c|c}
\pm 1 & 0\\ \hline
\av^t & \hat A
\end{array}\right):
\av\in \mathbb{R}^3, \hat A\in O(3,\R)\right\}.
$$
Thus,  the dual basis $B^*=(\phi^0, \phi^1,\phi^2,\phi^3)$ of each IFR contains the same first element $\phi^0$, up to a sign. When a choice in $\{\phi^0,-\phi^0\}$ is carried out, name it   $ \ttt :V\rightarrow \R$, then 
  $\ttt $  is called the {\em absolute time}. The kernel $E$ of $\pm\phi^0$ is endowed with a scalar product $\langle\cdot,\cdot \rangle_E$ (being  the elements $ (e_1,e_2,e_3)$   
  of $B$ an orthonormal basis of $\langle\cdot,\cdot \rangle_E$ for each IFR). Then, $E$ endowed with this scalar product is called the {\em absolute 
space}.

Summing up, the spacetime of Galilei-Newton is recovered now.

\item Case $k=0$. The group $\ok$ becomes the  {\em dual Galilean group}\footnote{This group was studied by L\'evi-Leblond \cite{carrollLevy}, who named it {\em Carrollian group}, after Lewis Carroll. Even though introduced as an academical exercise, recent applications of this group can be found in \cite{CarrollFigueroa, CarrollianGibbons}.}
$$
O^{(0)}(4,\mathbb{R}):=\left\{
\left(\begin{array}{c|c}
\pm 1 & \ah \\ \hline
0 & \hat A
\end{array}\right):
\ah \in \mathbb{R}^3, \hat A\in O(3,\R)\right\}.
$$
In this case, the  basis $B=(e_0, e_1,e_2,e_3)$ of each IFR contains the same first element $e_0$, up to a sign. Choosing a sign, this vector defines the  {\em absolute rest observer}. Thus, the kernel  (annihilator) of $\pm e_0$ in the dual space $V^*$ (that is,  the subspace $E^*:=$ Span$\{\phi^1,\phi^2,\phi^3\}$ of $B^*$ for each IFR) is also independent of the IFR. It is  naturally endowed with a scalar product 
 $\langle\cdot,\cdot \rangle_{E^*}$ so that,   for each IFR, the set $(\phi^1,\phi^2,\phi^3)$ becomes an orthonormal basis.

Summing up, an  a priori  aphysical dual of  Galilei-Newton spacetime (with a completely analogous geometric structure) is obtained.

\item Case $k\in(0,\infty)$. For $k=1$, the group $\ok$ is the Euclidean orthonormal group\footnote{It is worth pointing out that, in this case, not only the independent symmetries between time and spatial coordinates in formula \eqref{e_post2} hold, but also the crossed symmetries $\partial t/\partial \tilde x^i=\partial \tilde x^i/\partial t$ and $\partial x^i/\partial  \tilde t=\partial  \tilde t/ \partial x^i$ appear now.};
 otherwise, $\ok$ is conjugate to this group. Indeed, reasoning as in the case $k<0$, $V$ is naturally  endowed  with an Euclidean scalar product $\langle\cdot , \cdot \rangle_0$ and any basis  $B$ of an IFR is  orthornormal 
for $\langle\cdot , \cdot \rangle_0$, up to the normalization of its first vector.
 
Summing up, one obtains the  a priori aphysical case when  the full spacetime is endowed with  a Euclidean scalar product, which is mathematically analogous to the Lorentzian one.

\item Case $k\in S^1$ non-unique. In this case, the group is $G=\{\pm 1\}\times O(3,\R)$ and, thus, the basis $B$ and its dual $B^*$ for any IFR satisfy all the properties in the previous cases. In particular, choosing a sign, one has an absolute time $T$, an absolute rest observer $e_0$ (with $T(e_0)=1$) and an absolute space $(E,\langle\cdot,\cdot \rangle_{E})$ whose dual space can be identified  with $(E^*,\langle\cdot,\cdot \rangle_{E^*})$ defined in the case $k=0$.

This case should be regarded as aphysical too\footnote{Anyway,  it would represent the   model of space and time which goes back to Aristoteles. Recall that in that model, one would assume not only the existence of the absolute space and time but also that, for any $P\in  \Aff$, there exists a physical observer at $P$ at absolute rest. This would determine the affine line   $P+ \R \cdot e_0$, which would be regarded as a ``space point at any time''.}  and, being obtained as a ``degenerate'' case of the previous ones, it will not be taken into account anymore.   
\end{enumerate}

\subsection{Temporal models and interpretation of $k=-c^2$}\label{s2.4}
Taking into account the previous four models of spacetime which depend on a unique $k\in S^1$, let us revisit the role of the  hypothesis  of  {\em apparent temporality}  in Theorem \ref{t_k}. 


Recall that apparent temporality was enough to ensure the existence of $k$ in that theorem. However, the Euclidean case $k>0$ would not be excluded by this hypothesis, because the set $\Si$ of all the IFR's might contain  ``few'' elements (so that  only  transition matrices $A$ with $a_{00}=\cos\theta$ appeared for values of $\theta$ with $\cos \theta>0$). \ Moreover, in the other three cases for $k$, the elements of $\Si$ would determine a {\em 
time-orientation}\footnote{That is, a choice of one of the two timelike cones when $k<0$ and one of the two choices of absolute time or absolute rest observer when $k=\omega, 0$, resp.} under apparent temporality, but there would still be elements in $\Si^*$ which would not match with the chosen time-orientation.
However, when the case $k>0$ is disregarded a priori (say, regarding it as  aphysical), it would be natural to strengthen the hypothesis of apparent temporality into {\em temporality}, namely: all the  transition matrices between pairs of elements of $\Si^*$ in \eqref{e_sifrestrella}  have $a_{00}>0$.
This requirement not only would exclude the group $\ok$ for $k>0$ but it would also imply a restriction on the group 
for the other cases. This discussion  makes natural the following definition and convention. 
    
\begin{defi}\label{def_temp} The linear models of spacetime with $k\in 
 (-\infty,0) \cup \{\omega, 0\}$ will be called  {\em temporal models}. 
 When only these  models are considered, we will assume that {\em apparent temporality} also holds and, then,  the following {\em convention of temporality} can be assumed with no loss of generality: 

 (a) The  temporal models are  time-oriented. 

 (b) All the elements in  $\Si$ are assumed to lie in the chosen time-orientation.  

 (c) $\Si$ is assumed to be maximal  for the property (b). Thus, depending on the value of $k$, the {\em orthochronous} subgroup of the    Lorentz (or conjugate to Lorentz), Galilean or dual Galilean group will act freely and transitively on $\Si$. 

 (d) When there is no possibility of confusion, $\Si^*$ is regarded as equal to $\Si$ in (c). 
\end{defi}


For temporal models, given a transition matrix $A$  which gives the coordinates $(\bar t,\bar x^j)$  for $\bar R$ from the coordinates $(t,x^j)$ of  $R$,
 the {\em  velocity} and {\em speed} of $R$ measured by $\bar R$ are, resp., 
\begin{equation}\label{e_v}
\mathbf{{v}}=\av/a_{00} \qquad |\mathbf{v}|= \sqrt{\sum_{i=1}^3 (v^i)^2}
\end{equation}
in the notation of Exercise \ref{ex_1} (see also \cite[\S 5 (2)]{BLS}).

\begin{prop}
For any temporal model, $c:=\sqrt{|k|} \in [0,\infty]$ is the supremum of speeds measured between IFR's in  $\Si^*$.
\end{prop}
\begin{proof}
For $k=0,\omega$ this follows from \eqref{e_v} taking into account the expression of $\ok$ at each case  (see  items (2) and (3) at \S \ref{s2.3}). 
For $k\in (-\infty,0)$,  
using 
 \eqref{e_conj} the first column of $A$ is $(a_{00},\av)=(a_{00},cb^1,cb^2,cb^3)$ with
$\sum_i(b^i)^2=a_{00}^2-1$  (thus, $|\mathbf{v}|^2=c^2-c^2/a_{00}^2$) and
   $a_{00}\geq 1$  unbounded.
\end{proof}
As one would expect, this supremum is $\infty$ (i.e. the speeds are unbounded) in the Galilei-Newton case, finite equal to $c$ in the case of Special Relativity and strictly 0 in the dual Galilean case (where all IFR's lie at absolute rest). 

\begin{rem}\label{r_light} In principle, it is appealing to call $c$ the {\em speed of light}. Notice, however, that there is no mention neither to Electromagnetism nor to any other interaction in our approach. 
Nevertheless, an essential property of electromagnetism can justify that name. 
Namely, light is described by a wave {\em which propagates in vacuum}. An obvious natural hypothesis for IFR's is that the vacuum is ``equal'' for all of them, and, so, any physical scalar quantity measured  with respect to the vacuum must yield the same number for all of them. In particular, this would mean that all IFR's must measure the same speed of propagation of the light with respect to the vacuum. As the supremum $c$ is the unique speed equal for all of them, the following definition is justified.
\end{rem}
\begin{defi}\label{d_speedoflight}
For any temporal model, $c= \sqrt{|k|}$ is called the  {\em speed of light}.
\end{defi}
Anyway, the following digression about the physical content of this definition may be worthy. If one considered another interaction which also propagated in vacuum (say, gravitation) then the  arguments  in Remark \ref{r_light} would imply that its speed of propagation $c'$ with respect to vacuum would be the same $c$ as for light. 
As emphasized by some authors,  see \cite{Geroch},  there is no  logical contradiction   assuming that $c\neq c'$ and, thus, this question becomes  an experimental 
issue\footnote{However, recent measurements of gravitational waves show that the  speed of propagation of light and gravitation are equal with an extraordinary accuracy \cite{LIGO}.}. In the affirmative, these different interactions might allow one to construct different types of clocks and rods in order to measure the temporal and spatial coordinates. So, the name { 
	 IFR} should include the interactions which allow Postulates \ref{p1} and \ref{p2} to hold.

\section{First non-linearization}\label{s3}

General Relativity can be regarded as a first non-linear generalization of Special Relativity. Such nonlinearity comes from the fact that   Postulate \ref{p1}, namely, the global affine character of  spacetime, is being dropped and the set of all the events  is modeled  by a manifold. Nevertheless  (as apparent from \cite{BLS}),  Postulate \ref{p2} would make still sense if the symmetries stated there are regarded just as infinitesimal ones, at the tangent space of each event. 

 This  idea is well-established in the Lorentz case and it may seem very speculative in the other linear models of spacetimes.   However,  
this will be developed  briefly along this section  with a double aim: on the one hand, the  role of  observers will be emphasized and,  on the other,  the framework of  further issues relevant to the Lorentz-Finsler case will be   settled.  
Only in \S \ref{s4}  we will focus on the Lorentz case and will go beyond, in order to  reach  the Lorentz-Finsler generalization. 

\subsection{General case and signature change}\label{s3.1}
Assume now that the spacetime is described by a (smooth, connected) manifold $M$ and that our postulates 
are regarded as infinitesimal requirements of symmetry at the tangent space $T_pM$ of each $p\in M$, that is, around each event $p\in M$, one can find a set of coordinate charts  such that the relations \eqref{e_post2} occur only at $p$,  namely, considering normal coordinates. 

Then, we will have a set   $\Sm_p$ of linear bases at  each $T_pM$ which  will  play the role of (linear) IFR's at $p$. For simplicity, we will assume in what follows:

 (i)   $\Sm_p$ determines univocally some $k(p)\in S^1$ (i.e., the degenerate case  of non $k$-congruent solutions of Exercise \ref{ex_1}  is skipped), 

 (ii)  $\Sm_p$ is maximal (i.e.,  $\Sm_p=\Sm_p^*$,  consistently with the discussion above Theorem~\ref{t_k}),  and

(iii)  Consistently with  Definition \ref{def_temp}, the convention  of temporality will be assumed whenever $k(p) \not\in (0,\infty)$ (in particular, the  notion of future-directed timelike vectors  makes sense  then).

 Moreover, as an  extra hypothesis (or   third  postulate, as in \cite{BLS}) we assume:
\begin{itemize}
\item[(P3)] $\Sm_p$ varies smoothly in the bundle $LM$ of linear frames\footnote{$LM$ contains all the (ordered) linear bases of $T_pM$ for all  $p\in M$.} of $M$.
\end{itemize} 
Formally,  this means that $\Sm := \cup_{p\in M}\Sm_p$ is a smooth bundle embedded in $LM$
 (in the sense of a submanifold of $LM$ with the induced topology such that the  projection on $M$ is a submersion)   so that the function $k : M \rightarrow  S^1 $ becomes smooth.  

In general, one obtains then a signature changing metric $g$ which is Lorentzian  (resp. Riemannian)  in the  set $-\infty<k<0$ (resp. $0<k<\infty$). Following the terminology in  \cite{BLS, BS_jmp}, in the closed subset determined by 
$k=\omega$, one has a {\em Leibnizian structure}, that is, a non-vanishing  1-form  $\Omega$ ({\em absolute time form}) on $M$  and a Riemannian metric $h$ in the subbundle ker$(\Omega)$ of $TM$, being then $($ker$(\Omega), h)$ the {\em absolute space}\footnote{Such a structure is equivalent to having the 1-form $\Omega$ and a positive semidefinite 2-contravariant tensor $T$ of rank 3 with 
$i_\Omega T (:=T(\Omega,\cdot )) \equiv 0$, studied in \cite{Kunzle}. Indeed, such a $T$ induces a Riemannian metric in the dual of ker$(\Omega)$ and, then, in ker$(\Omega)$. Conversely, the Leibnizian structure yields a Riemannian metric on ker$(\Omega)$ and then in its dual; this yields the tensor $T$ by imposing that its radical is Span$\{\Omega\}$.}. Analogously, the region $k=0$
is endowed with an anti-Leibnizian structure, consisting in a non-vanishing vector field $W$ ({\em absolute rest field}) on $M$ and a Riemannian metric $h^*$ on the subbundle  ker$(W)$ of the cotangent bundle $TM^*$.

Let $g$ be  the semi-Riemannian (Lorentzian or Riemannian) metric  in the region $k\neq 0,\omega$ and $g^*$ the (physically equivalent) metric induced in the cotangent space. It is worth emphasizing that, in the region $k=0$, $g$ can be extended as a degenerate metric 
and $g^*$ cannot; however, $g^*$ matches smoothly with $h^*$ on ker$(W)$. Analogously, in the region $k=\omega$, $g^*$ can be extended as a degenerate metric, while $g$ matches smoothly with $h$ on ker$(\Omega)$.

Summing up, this first non-linear generalization of the IFR setting yields as a general model of spacetime a geometry governed by the smooth function $k$. Whenever $k\neq 0,\omega$,  a semi-Riemannian metric $g$ and its equivalent dual metric $g^*$ are obtained; in the regions $k=0$ or  $k=\omega$ either the metric $g$ or $g^*$ are extended as a degenerate metric   
and  additional geometric structures appear\footnote{Recall that models of  signature changing metrics have been studied at least since the influential  ``no boundary'' proposal by Hartle \& Hawking \cite{HH}, see for example \cite{Dray,WWV}. Moreover, the existence of an ``absolute time'' in the transition region has also been pointed out by several authors \cite[\S 2]{KK} (see also \cite{VJ}).}. 
The transition among these elements is smooth, as so is $\Sm$. 

\subsection{Space of observers} For convenience, let us introduce a new element by taking the most important information from $\Sm$.
\begin{defi}\label{d_So}
The {\em space of observers}  is the subset $\So$   of $TM$ containing the first vector of each basis at $\Sm$, and the     {\em space of observers at $p$} is $\So_p:= \So \cap T_pM$. That is:
\begin{enumerate}[(i)]
\item In the region  $k>0$, $\So$ contains all the  unit vectors for $g$,  and 
each $\So_p$ is a sphere. 

\item In the region $k<0$, $\So$ contains all the future-directed timelike unit vectors for $g$,  and  each $\So_p$ is a hyperboloid.  

\item In the region $k=\infty$, $\So$ is equal to $\Omega^{-1}(1)$,  and  each $\So_p$ is an affine hyperplane not containing $0$.

\item In the region $k=0$, $\So$ is equal to the  absolute rest vector field $W$ (so, each $\So_p$ is a subset containing a single non-zero tangent vector).
\end{enumerate}
\end{defi}
\begin{rem} If $n=$ dim $M$, then dim $LM= n(n+1)$  and $\Sm \subset LM$ is always a submanifold with dim $\Sm=(n+1)n/2$. Nevertheless,
$\So$ must be regarded as a subset of $TM$. Then, it becomes a smooth manifold of dimension $2n-1$ in the region $k\neq 0$ but   it collapses to a submanifold of dimension $n$ when  $k=0$.
\end{rem} 
The transition from Lorentzian to 
Riemannian through a region with $k=\omega$ can be easily understood by looking at $\So$  (see Fig. \ref{Fig1}).  We will not be  interested in the transition through a region with $k=0$. However, this could be described in a completely analogous way by defining a dual space of observers (constructed by picking at each point $p$ the first element of the  elements in the bases which are dual to those in $\Sm_p$). 
\newsavebox{\smlmat}
\savebox{\smlmat}{$\left(\begin{array}{c|c}k & 0\\
 	\hline 0 & 1\end{array}\right)$}
 	\newsavebox{\smlmatt}
\savebox{\smlmatt}{$\left(\begin{array}{c|c}1/k & 0\\
 	\hline 0 & 1\end{array}\right)$}
\begin{figure}
	\centering
\begin{tikzpicture}
\draw (-6,0) -- (6,0);
\draw[dashed] (-6.1,4.5) -- (-4,0) -- (-1.9,4.5);
\draw[thick, color=blue] (-6,4.5) .. controls (-4,1) .. (-2,4.5);
\draw[->] (-4,0) -- (-4,1.87); 
\node at (-4,-0.5)(s){$p_{-1}$};
\node at (-4,-1.2)(s){$k=-1$};
\node at (-4,2.2)(s){$e_0$};
\node at (-3.2,3.5)(s){$O_{p_{-1}}$};
\draw[thick, color=blue] (-2,2) -- (2,2);
\draw[->] (0,0) -- (0,2);
\node at (-0,-0.5)(s){$p_\omega$};
\node at (0,-1.2)(s){$k=\omega$};
\node at (0,2.2)(s){$e_0$};
\node at (1,2.5)(s){$O_{p_{\omega}}$};
\draw[thick, color=blue] (4,0) circle(1.6);
\draw[->] (4,0) -- (4,1.6);
\node at (4,-0.5)(s){$p_{1}$};
\node at (4,-1.2)(s){$k=1$};
\node at (4,1.87)(s){$e_0$};
\node at (5,1.87)(s){$O_{p_{1}}$};
\end{tikzpicture}
 \caption{\label{Fig1} Signature changing spacetime on $M=\R^2$. The natural coordinate basis $B=(e_0,e_1)\equiv (\partial_t,\partial_x)$ is assumed to lie on $\Sm$ at each point. The matrices of the metric $g$ and $g^*$ are~\usebox{\smlmat} and~\usebox{\smlmatt},
 	 respectively, with $k(t,x)=1/x\in S^1\setminus\{0\}$. 
 	The space of observers changes from a hyperbola to a line and to a circumference.
 }
 \end{figure}
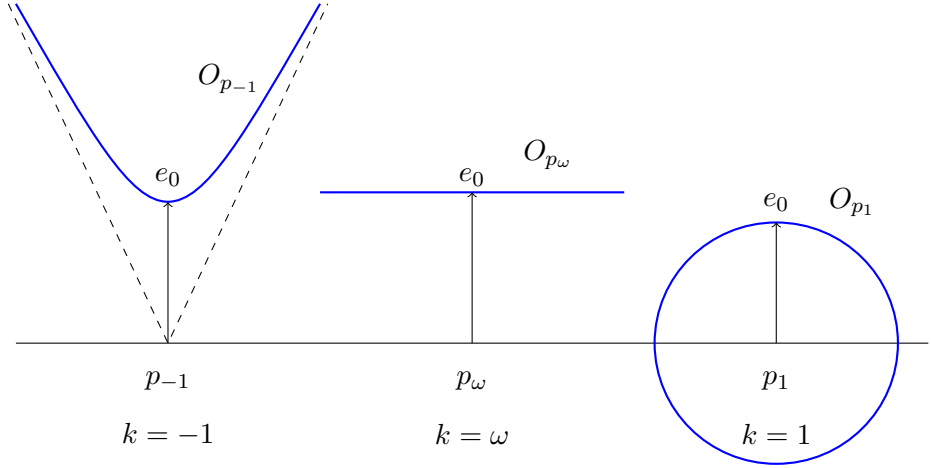
%


\subsection{ Pointwise variation of speed of light }\label{s3varying_c} In the region $-\infty <k <0$ the function $c(p)=\sqrt{|k(p)|}$ might be understood as a variation of the speed of light with the point at $M$. 
Such a possibility has been speculated since the  beginning of General Relativity, and was put forward in the 90's in relation to cosmological inflation and the horizon problem  (see for example 
\cite{Magueijo, Barrow,  Moffat, Petit},  as well as some criticism in \cite{Ellis, Uzan}). In order to avoid the circularity of using the light to define the 
units to measure its speed,   Barrow and Magueijo  \cite{BarrowM}
argue that only the variation of  adimensional constants would have a true physical meaning; so, the variation of $c$ should be regarded as a variation of the   (adimensional)  fine structure constant  $\alpha=e^2/\hbar c 4\pi \epsilon_0$. 


 Without deepening into these questions,  some comments about varying $c$ in our framework are in order. 
Recall first that, in the affine case obtained by assuming  Postulate \ref{p1},  to assume also  Postulate \ref{p2} would imply that all the IFR's  would be  using the same units for measurement and, then, the specific value of $c$ would depend of the chosen units. Indeed, the natural interpretation of the group $O^{(k)}(4,\R)$ for $k=-c^2\in (-\infty,0)$ is just the Lorentz group in some appropriate coordinates.
Thus, if one regarded the affine space $\Aff$ as a manifold and took different units at each point, then this could not be interpreted as a variable speed of light. 

 To measure   a varying speed of light  would rely on the
possibility to compare the units of measurement at different points.
In the affine case such a comparison would be possible if the interactions were invariant by translations (an unlikely possibility). 
In general, one would need  measurements involving magnitudes which are dimensionally independent  (in particular,   this would be achieved by measuring adimensional constants, as commented above). In principle, this might be achieved by
measuring essentially different interactions, as in the case of light and gravity propagation (see below  
Definition~\ref{d_speedoflight})\footnote{In the more speculative case of  a signature changing metric, the speed of light would change necessarily in the regions $k=0,\omega$. So, the possibility to measure a varying speed of light when $-\infty < k < 0$ would imply that the collapse  of the lightcones (to a line or a hyperplane)  could be measured gradually when approaching to those regions.}. 
Anyway, as we will see, the  Lorentz-Finsler viewpoint will open other  possibilities by using infinitesimal anisotropy.

\subsection{Relativistic vs Leibnizian structures}\label{s3.4}  A priori, the cases  $-\infty<k<0$ and $k=\omega$ (with constant $k$, and then $c$) are the physically  interesting ones, either as a model of spacetime or as an approximation to this model. Next, they will be briefly compared.

In the first case,  we will assume  $k=-1$ on all $M$, after our discussion in the previous subsection. So, one has a time-oriented Lorentzian metric $g$ and $\So$ is a fibered space on $M$ with fiber the hyperbolic space. Moreover, the Levi-Civita connection $\nabla^g$ is canonically associated with $g$, and any other affine
connection $\nabla$ parallelizing $g$ (i.e., satisfying $\nabla_\mu g_{\nu\rho}=0$) must be non-symmetric. 
This selects $\nabla^g$ and provides a sense of free fall and lightlike geodesics (compare with EPS later).

In the case $k=\omega$, the Leibnizian structure consisting in the absolute time form $\Omega$ and the absolute space $($ker$(\Omega), h)$  on $M$  described in subsection \ref{s3.1} is obtained. These structures were studied systematically in \cite{Kunzle} and \cite{BS_jmp}. 
In this setting, one  considers {\em Galilean connections}, that is,  affine  connections $\nabla$ which parallelize both, $\Omega$ and $h$.
It is worth pointing out that the set of all the Galilean connections has the same degrees of freedom as the set of all   affine  connections (symmetric or not) parallelizing a Lorentzian metric $g$. However, a symmetric Galilean connection will exist if and only if $\Omega$ is closed ($d\Omega=0$), that is, locally  $\Omega=d\ttt$ for some function $\ttt$. Nevertheless, in this case, there is no univocally determined symmetric connection. 
Moreover, this happens even if $\Omega$ is exact, that is, $\Omega= d \ttt $ for some function $\ttt$  defined globally on $M$, which will be  called the {\em absolute time} ($\ttt$ is unique up to an additive constant).   Indeed, an explicit 
Koszul-type formula  reconstructs all the symmetric Galilean connections in terms of two data\footnote{Such a formula can be extended to include non-symmetric connections by adding as a third datum a suitable component of the torsion, see  \cite[Th. 27]{BS_jmp}.} \cite[Cor. 28]{BS_jmp},  namely: the {\em gravitational  field} 
(a vector field in the absolute space, that is, a section of the bundle ker$(\Omega) \rightarrow M$) and the {\em vorticity}  
(a skew symmetric 2-form on the vector bundle ker$(\Omega)$). 

 In conclusion,  relativistic spacetimes are preferred to Leibnizian ones from the viewpoint of foundations, because  of two celebrated properties: 
(a) they  permit to model a finite speed of propagation in vacuum  (recall that  observers appear now at each event as infinitesimal approximations to IFR's and, so, the arguments in Remark \ref{r_light} apply), and 
(b) they select a unique affine connection in the set of all the   connections parallelizing the geometric structure, while Leibnizian ones require  the gravitational and vorticity fields as an extra input.

 In the next section, we will focus only on relativistic spacetimes and the Finslerian extensions. Nevertheless, some  previous elements 
 serve 
as a background for the Lorentz and Lorentz-Finsler cases and they can be compared a posteriori with them 
 (see
 Table~\ref{tab1}). 
We point out a pair of them so that the interested reader might come back here later: 
 
 (1) Leibnizian structure $(\Omega,h)$ vs  cone triple $(\Omega,T,F)$ (which is useful to define and to handle any cone structure $\C$, Lorentz or Lorentz-Finsler, see Definition \ref{d_cone}, Remark \ref{r_cone}(\ref{r_cone2})).   Notice that when $F$ comes from a Riemannian metric $h$, then the Leibnizian structure can be regarded as a sort of limit when $\lambda\rightarrow \infty$ of the triples $(\lambda \Omega,T/\lambda,F)$, which ``open'' the cone $\C$.
 
 (2) Chronometric vs EPS approach to spacetime (\S \ref{s5}). The Leibnizian structure   $(\Omega,h)$ (eventually, with $\Omega=d\ttt$)  gives a chronometric approach to spacetime, in a similar way as the Lorentz metric does in Relativitiy. However, the former requires an additional input (an affine connection)  in order to define free fall. So, the EPS approach (at least the axioms which do not consider light propagation) might  also be interesting in the Leibnizian case.    
In contrast,   the Lorentz-Finsler metric $L$ will provide timelike and lightlike geodesics in a very similar way as the Lorentz metric $g$, in spite of the differences between the Levi-Civita $\nabla^g$ and the anisotropic connection  (see footnote \ref{foot_anisotropic})  determined by $L$.

\begin{note} {\em Newton-Leibniz controversy}. To end this section it is worth pointing out that the notion of Leibnizian structure provides a precise mathematical description of  a historical controversy between Leibniz and Newton. Roughly speaking, Leibniz criticized  Newton's arguments about  IFR's 
 by pointing out that  the Euclidean space perceived  by an observer is equal even after a rotation of the observer's coordinates. So, he claimed that one could not detect whether these axes are being rotated at different times. Newton replied that  spinning water in a bucket would detect whether the observer is rotating or not. From the mathematical viewpoint, Newton was using the structure of a  Galilei-Newton spacetime, as described in \S \ref{s2.3} (that is, the  linear quadratic classical space + time approximation in Table~\ref{tab1}). So, the overall affine structure of the (four dimensional) spacetime yields a natural affine connection, which can be used to detect rotation.  Leibniz, however, is considering physical spacetime only as a manifold endowed with a Leibnizian structure  (that is, he drops the spacetime affine structure and considers only the pointwise quadratic  first nonlinear space + time in Table~\ref{tab1}). So, with these elements, no affine connection is determined, and rotation cannot be measured.
    Summing up, Leibniz was right pointing out that, only with the Leibnizian structure  on $M$ at hand, no Galilean connection is selected\footnote{In our opinion, this justifies the name {\em Leibnizian} used here (following \cite{BS_jmp}), compare with \cite{Kunzle, Sancho}.}. However, Newton  did select such a connection by guessing the further affine structure of~$M$.
\end{note}


{\footnotesize
\begin{table}
	\centering
	\begin{tabular}{|m{1.5cm}|m{2.4cm}|m{2.4cm}|m{2.4cm}|m{2.4cm}|}
		\hline
		\multirow{2}{1.5cm}{MODEL}  & \multicolumn{2}{m{5cm}|}{{\bf Linear space: affine Aff with vector $V$} 
(translation-invariant elements)
Geodesics $\equiv$ straight lines} & \multicolumn{2}{m{5cm}|}{{\bf Smooth connected manifold M}
(pointwise dependent elements)} \\ \cline{2-5}
& Quadratic forms
(doubly  linear) & No quadratic 
restriction
& First nonlinea\-rizat.: pointwise quadratic & Second nonlinea\-rizat.:  no  quadratic restriction
\\ \hline

SPACE &  {\bf\footnotesize Euclidean scalar 
product $g_E$ on $V$} 

\vspace{0.5cm}
Symmetry $O(n)$ & {\bf Minkowski norm $\|\cdot\|$}

\vspace{0.5cm}
(drop parallelogram identity + reversibility) &
{\bf Riemannian metric $g_0$}

\vspace{0.5cm}
Unit sphere bundle = pointwise ellipsoid.
Levi-Civita natural mathematical choice. & {\bf Finsler metric $F (L=F^2)$}

\vspace{0.5cm}
Indicatrix = 
pointwise strongly convex hypers.
Cartan connection \\ \hline

Space + time Classic &

{\bf Galilei-Newton
$(\ttt, g_E)$ }

\vspace{0.5cm}
-$\ttt$ absolute time 
(on $V$): Non-zero linear form 

\vspace{0.5cm}
-$({\rm Ker}\, \ttt, g_E)$ absolute space:  $g_E$  Eucl. scalar product on ${\rm Ker}\, \ttt$.

\vspace{0.5cm}
Symmetry: orthochr. Galilean group &
{\bf Non-quadratic Galilei-Newton
$(\ttt, \|\cdot\|)$ }

\vspace{0.5cm}
-Replace 
 $g_E$ in Galilei-Newton by a norm.
 
 \vspace{0.5cm}
Not developed (as far as we know) &
{\bf Leibnizian structure:}

\vspace{0.5cm}
Non-vanishing $1$-form $\Omega$ (eventually, $\Omega=d\ttt$) with a Riemannian metric on the bundle ${\rm Ker}(\Omega)$.

\vspace{0.5cm}
Required to choose a linear connection parallelizing $\Omega$ and $g_R$) &
{\bf Leibniz-Finsler str. }

\vspace{0.5cm}
-Replace the Riemannian metric on 
${\rm Ker}(\Omega)$ by a Finslerian one.

\vspace{0.5cm}
Not developed (as far as we know) \\ \hline

Space
-time
Relat. &

{\bf Special Relat. $(g_L, C)$ }

\vspace{0.5cm}
-  $g_L$  Lorentzian scalar product $(+,...,+,-)$

\vspace{0.5cm}
- $C$ time-orientation
(choice of one between 2 cones)

\vspace{0.5cm}
Symmetry : orthochr. Lorentz group $O^\uparrow_1(n)$  &
{\bf Modified Special Relat. $(L_0,\C_0)$}

\vspace{0.5cm}
-$C_0$ cone 

\vspace{0.5cm}
-$L_0$ Lorentz-norm on $\C_0$ causal vectors
(eventually $\C_0$ determined from $L_0$)

\vspace{0.5cm}
No symmetry but includes the case VSR  (with proper a subgroup of  $O^\uparrow_1(n)$ ) &

{\bf 
	General Relat. $g_1$: }

\vspace{0.5cm}
Pointwise smooth Lorentzian scalar product $g_1$ continuously time-oriented 

\vspace{0.5cm}
Levi-Civita connection: free fall, ligthlike pregeodesics, gravitational force &
{\bf Finsler spacetime $L$ (with a cone str. $\C$)}

\vspace{0.5cm}
Pointwise smooth Lorentz-Minkowski norm. 

\vspace{0.5cm}
Geodesics determined by Cartan (and Chern etc.) connection.

\vspace{0.5cm}
$\C$ pregeodesics independent of $L$ 

\vspace{0.5cm}
Anisotropy with causal directions (possibly due to matter/energy)
\\ \hline
	\end{tabular}
	\caption{\label{tab1} Classical models of non-quantum space and time (linear models and their  non-linearizations) }
\end{table}
}

\section{Second non-linearization}\label{s4}

\subsection{Background: norms, cones and Lorentz-Finsler metrics}\label{s4.1} In order to show rigorously the emergence of the notion of Finsler spacetime, some purely  geometric elements are stressed first. Even though some of them are elementary, they will be necessary to make precise discussions. So, the experimented reader can skip some parts and come back when necessary.

The first ones come from classical norms on a (finite-dimensional, real) $n$-vector space $V$ and Finsler Geometry; they are carefully explained  in \cite{JavSan11}.
\begin{defi}\label{d_minkowski} A {\em Minkowski norm} on $V$ is a map $F_0:
V\rightarrow \R$ satisfying
\begin{enumerate}[(i)]
\item positiveness: $F_0(v)
\geq 0$, with equality  if  and only
if $v=0$,
\item positive homogeneity: $F_0( \lambda v ) =
\lambda F_0(v)$ for all  $\lambda > 0$,
\item strongly convex indicatrix: $F_0
$ is
smooth 
away from $0$ and the {\em fundamental
tensor field} $g$ defined as
the Hessian of $\frac 12 F_0^2
$ 
is positive definite  on $V\setminus\{0\}$. 
\end{enumerate}
\end{defi}
\begin{rem}\label{r_minkowski}
Notice about this definition:

\ben\item {\em Positive homogeneity}. This requirement only for $\lambda>0$ enhances the applications of Finsler Geometry\footnote{Including those for relativistic stationary spacetimes, see \cite{FHS} and references therein.}, and it will be enough for our purposes.
Positive homogeneity implies that $F_0$ is univocallly determined by its {\em indicatrix} (unit sphere) $\Sigma_0:=F^{-1}
(1)$. In particular,  the full homogeneity of $F_0$  becomes equivalent to the symmetry of $\Sigma_0$ with respect to the origin. 

\item \label{r_minkowski_item2}{\em Smoothness}. The standard definition of norm implies that they are only continuous.  We assume smoothness (say, $C^\infty$,  pointing out the cases when lower regularity becomes relevant)   away from 0.
Using (ii),  this is clearly equivalent to the smoothness of $\Sigma_0$.   


\item\label{r3_def4.2} {\em Role of  triangle inequality}. It is not imposed directly, however:
\ben\item
Triangle inequality becomes equivalent (for any 1-homogeneous function smooth away from 0) to the {\em convexity of $\Sigma_0$} (i.e., its inner-pointing second fundamental form $\sigma$ is positive semidefinite). Moreover, it is  also equivalent to the {\em convexity of the open unit ball} $B_0:=F_0
^{-1}([0,1))$ (all the segments connecting points $u,v\in B_0$ are included in $B_0$).

\item The strict triangle inequality becomes equivalent to the {\em strict convexity of} $\Sigma_0$
(the hyperplane tangent to $\Sigma_0$ at each point $p$ only  intersects $\Sigma_0$  at $p$).
 Moreover, it is  also equivalent to the {\em strict convexity of the closed unit ball} $\bar B_0$ (segments connecting points $u,v\in \bar B_0$ are included in the open ball  $B_0$ up to the endpoints $u,v$).

\item  Assuming (i) and (ii),  the hypothesis  (iii) becomes equivalent to the {\em strong convexity of  $\Sigma$} ($\sigma$ is positive definite),  which   is more restrictive than its strict convexity. 
\een

\item\label{r_minkowski4} {\em Conic Minkowski norms}. These norms are  as in Definition \ref{d_minkowski} just by allowing the map $F_0
$  to be defined only on a cone domain (see  Definition \ref{d_scone}  below) $A_0$ of $V$. All the
previous considerations on the triangle inequality extend trivially to  such  conic Minkowski norms.
 \item \label{r_minkowski5} {\em Scalar products}. Norms coming from  (Euclidean) scalar products are Minkowski. Conversely, a Minkowski norm comes from a scalar product under one (and, then, both) of the following properties:
\ben\item The classical parallelogram identity holds.
\item $F_0^2
$ is  $C^2$-smooth at zero \cite[Proposition 4.1]{Warner65}.
\een
Recall also that, clearly, any norm coming from a (Euclidean or Lorentzian) scalar product is  determined by its value on a cone domain. 
\een
\end{rem}

\begin{defi}\label{d_finslerm}
A {\em Finsler metric} $F$  on a manifold $M$  is a function $F:TM\rightarrow  \R$ satisfiying: (i) $F$ is smooth away from the zero section $\mathbf{0}\subset TM$ and (ii) the restriction $F_p$ of $F$ to each tangent space  $T_pM$, $p\in M$, is a Minkowski norm.
 
\end{defi}
\begin{rem}\label{r_finslerm}
Notice about this definition:

\ben\item {\em 2-homogeneity}. Taking $F^2$ instead of $F$, Finsler metrics can be defined alternatively as  positive  2-homogeneous functions (this will be convenient for their Lorentzian  extensions). What is more, then the  $C^2$-smoothability of $F^2$ at 0  would imply  that it comes from a Riemannian metric (recall Remark \ref{r_minkowski} \eqref{r_minkowski5}). 

\item\label{r2_def4.2} {\em Role of the indicatrix}. As $F$ is determined by its  indicatrix $F^{-1}(1)$, 
a Finsler metric can be  defined alternatively as a  smooth hypersurface $\Sigma$ embedded in $TM$ satisfying appropriate conditions, namely:  (a)~$\Sigma$ intersects transversely\footnote{\label{foot_trans}For the role of the condition of transversality  see \cite[Prop. 12]{CJS14} and \cite[Def. 2.7, Rem. 2.8]{JavSan18}.} each $T_pM$ and (b)~this intersection  $\Sigma_p:= \Sigma \cap T_pM$ is a strongly convex compact connected  embedded  hypersurface whose inner domain $B_p$  (such that $\Sigma_p= \partial B_p$, where $\partial$ denotes the  boundary in $V$)\footnote{Such a $B_p$  exists by Jordan-Brower theorem.} contains the zero vector $0_p$. 

\item {\em Fundamental tensor on   a  vector bundle}. Each $F_p$ defines a fundamental tensor field on $T_pM\setminus \{0\}$ and, so, a 2-covariant tensor on each fiber of the (slit) tangent bundle $\pi: TM\setminus \mathbf{0} \rightarrow M$. We will use  the letter $g$ to denote such a tensor field, so that, for each $v\in TM\setminus \mathbf{0}$, $g_v$ will be a tensor on $T_pM$, being $p=\pi(v)$. Clearly, the definition of Finsler metric and fundamental tensor can be extended to any vector bundle, not necessarily the tangent one.   
\een
\end{rem}

\noindent  The rest of elements involves the Lorentz-Finsler case, and   we follow
\cite{JavSan14}. 
We start with the definition of cone. For our purposes, the next one is enough.  A more intrinsic  definition 
can be seen in \cite[Def. 2.1]{JavSan18}
(the equivalence and related optimal assumptions are analyzed in \cite[\S 2.1]{JavSan18}).  

\begin{defi} \label{d_scone} A (strong) {\em cone} $\C_0$ in  $V$ is any embedded hypersurface which can be constructed as follows: choose a hyperplane $\Pi\subset V$ which does not contain 0 and a strongly convex compact connected  embedded  $(n-2)$-hypersurface $S_0\subset \Pi$, take all  the  half-lines  from $0$ to  the points   of $S_0$ and define $\C_0$ as the union of all these half-lines except $0$. 

Then, the {\em  cone domain} is the open subset $A_0$  obtained analogously by taking the (open) half-lines from $0$ to each point of the inner domain $B_0$ of $S_0$ in $\Pi$ (so that  the boundary of $A_0$ in $V\setminus\{0\}$ is $\C_0$).   \end{defi}

\begin{defi} \label{d_cone} 
A {\em cone structure} on a manifold $M$ is a smooth embedded hypersurface $\C\subset TM$ such that:
 (a) $\C$ intersects transversely each $T_pM$ and (b) this intersection  $\C_p:= \Sigma \cap T_pM$ is a cone on $T_pM$. 
 Then, the {\em cone (structure) domain} is $A=\cup_{p\in M}A_p$, where each $A_p$ is the cone domain of $\C_p$. A vector $v\in A$ (resp. $v\in \C$; $v\in A\cup \C$) is called {\em timelike} (resp. {\em lightlike}; {\em causal}). 
\end{defi}

\begin{rem}\label{r_cone}
Notice about this definition:

\ben
\item\label{r_cone1} {\em Smoothness and transversality.}
Intuitively, a cone structure is just to smoothly put a cone at each $T_pM$, $p\in M$. From the formal viewpoint, however, this cannot be  deduced only from the smoothness of $\C$, making necessary assumption (a) (recall Rem. \ref{r_finslerm} (\ref{r2_def4.2}) and footnote \ref{foot_trans}, or the discussion around \cite[Fig. 2]{CJS14}).

\item\label{r_cone2} {\em Cone triples.} Any cone structure $\C$ can be determined (in a highly non-unique way) by means of a  cone triple $(\Omega, T, F)$, where $\Omega$ is any {\em timelike} 1-form on $M$ (i.e. $\Omega(v)>0$ for all causal $v$), $T$ is any timelike vector field with $\Omega(T)\equiv 1$ and $F$ is the unique Finsler metric on ker$(\Omega)$ such that $F(w)T+w\in \C$ for any $w\in $ ker$(\Omega)\setminus \mathbf{0}$ \cite[\S 2.4]{JavSan18}. Conversely, any   $(\Omega, T, F)$ with $\Omega(T)\equiv 1$ and $F$ Finsler on ker$(\Omega)$ is the cone triple of some cone structure $\C$.

\item\label{r_cone3} {\em Extended classical Causality.} $\C$ allows one to extend basic elements of Causality of spacetimes such as the chronological $\ll$, strict causal $<$, causal $\leq$ and horismotic $\rightarrow$ relations ($p\rightarrow q$ when $p< q$ and $p\not\ll q$) and, thus, the chronological/causal futures and pasts of a point,  $I^+(p), I^-(p)$ / $J^+(p), J^-(p)$. In particular, {\em cone geodesics} are defined as locally  horismotic curves, and they generalize the future-directed lightlike pregeodesics 
associated with the conformal structure of any Lorentz metric. 
\een
\end{rem}
 In the following, we will say that a function is smooth in a manifold with boundary (contained in a regular manifold $M$)  if it can be (locally) extended to a smooth function on an open subset of $M$.  
\begin{defi}\label{d_LMnorm} Let $\C_0$ be a cone on $V$ and $\bar A_0=\C_0\cup A_0$.
A {\em (properly) Lorentz-Minkowski norm} with cone $\C_0$ is a smooth map $L_0: \bar A_0 \rightarrow \R$ satisfying:
\begin{enumerate}[(i)]
	\item $L_0(v)\geq 0$ with equality if and only if $v\in \C_0$.
\item $L_0$ is positive 2-homogeneous: $L_0(\lambda v)=\lambda^2v$ for all $v\in \bar A_0$ and $\lambda>0$.
\item The fundamental tensor $g$ obtained as the Hessian of $\frac{1}{2}L_0$ has Lorentz\-ian signature $(+,-,\dots ,-)$ on $A_0$.
\end{enumerate}
\end{defi}

\begin{rem}\label{d_LMnormrem}
Consistently with the positive definite case,  let us observe the following:  

\ben \item A less redundant definition for $L_0$ (as well as for the Lorentz-Finsler metric $L$ below) can be carried out without prescribing the cone $\C_0$, see   \cite[Def. 3.1, 3.5]{JavSan18}.

\item Two homogeneity for $L_0$ is preferred to $1$-homogeneity  because of the general equality $L_0(v)=g_v(v,v)$. Notice  also that the Lorentzian signature is changed with respect to previous sections and, consistently, if $L_0$ is smoothly extended around any $v\in \C_0$, then $L_0$ 
must become negative away from $\bar A_0$.

\item\label{d_LMnorm_i3} $L_0$ is determined by its indicatrix $\Sigma_0=L_0^{-1}(1)$, which is now {\em strongly concave} and {\em asymptotic to} $\C_0$. Indeed, a Lorentz-Finsler metric could be defined alternatively as a strongly concave hypersurface $\Sigma_0$ in $A_0$ which is asymptotic to some cone structure $\C_0$ under the mild technical condition that the map $A_0\ni v\mapsto L_0(v)$ such that $v/L_0(v)\in \Sigma_0$, extend smoothly to $\C_0$ with non-degenerate\footnote{These conditions would be satisfied by  hypersurfaces  suitably $C^2$-close  to the space of observers $\So$ of any relativistic spacetime  (notice that some issues appear involving the extendability of $L$ to the cone and whether the cone is prescribed or not),  and they can be constructed for any cone (recall Rem. \ref{r_finslerspacetime}(\ref{i_4}) below).} $g$.

\item \label{d_LMnorm_i4} All the properties related to the triangle inequality in the positive definite case (which were associated with the convexity of the indicatrix and held for conic Minkowski norms, Remark \ref{r_minkowski} (\ref{r3_def4.2})) are automatically  translated now  as reverse triangle inequalities in the Lorentz-Finsler case (associated with the concaveness of $\Sigma_0$).

\item Even though  $\bar A_0 \subset V\setminus\{0\}$, $L_0$ can be continuously extended to $0$ ($L_0(0)=0$). However,  the smoothness of this extension\footnote{ Recall that, for any function $L_0$ on $\bar A_0\cup \{0\}\subset V$ (with $A_0$ a cone domain), the elementary definition of existence of a  differential map at $0$ makes sense  because $0$ is an  accumulation point of the domain   of $L_0$ and its uniqueness is guaranteed because  $A_0$ contains $n$ independent directions converging to 0.} depends on  whether $L_0$ comes from a Lorentzian scalar product, as in the positive definite case. 
\item  It is possible to smoothly extend $L$ (preserving  the $2$-homogeneity) to an open conic subset $A^*_0$ which contains $\bar A_0$ (recall that $0\notin \bar A_0$). This extension is far from unique, but the fundamental tensor in the boundary is well-determined. 
\een
\end{rem}
\begin{defi}\label{d_finslerspacetime} Let $\C$ be a cone structure on $M$ and $\bar A=\C\cup A$.
A {\em  (properly) Lorentz-Finsler metric} with cone $\C$ is a smooth map $L: \bar A \rightarrow \R$ satisfying that the restriction  $L_p$ of $L$  to each $T_pM \cap \bar A$ is a Lorentz-Minkowski norm. Then, $(M,L)$ is a {\em
  (properly)  Finsler spacetime}.\end{defi}

\begin{rem}\label{r_finslerspacetime}
The following results on Finsler spacetimes will be relevant:

\ben\item\label{i_1} Any Lorentz-Finsler metric can be extended to $TM\setminus \mathbf{0}$ as a smooth 2-homogeneous function with fundamental tensor $g$ of Lorentzian signature, see \cite{Min16}. However, such an extension is highly non-unique and, as we will see, it is not justified by direct measures of observers.

\item\label{i_2} Given $L$, timelike and lightlike geodesics are naturally defined and they satisfy local maximizing properties which extend  those of relativistic spacetimes (recall Rem. \ref{d_LMnormrem}(\ref{d_LMnorm_i4})  and \cite[Prop. 6.5]{JavSan18}).  In particular, the lightlike pregeodesics of $L$ coincide with the cone geodesics of $\C$ \cite[\S 6.2]{JavSan18}. 

\item \label{i_3} Thus, all the Lorentz-Finsler metrics  with the same cone structure have the same lightlike pregeodesics. Two such metrics $L_1, L_2$ are called {\em anisotropically equivalent} and they satisfy $L_2=\mu L_1$ for some $0$-homogeneous function $\mu>0$ on $\bar A$ \cite[\S 3.3]{JavSan18}.   

\item\label{i_4} Any cone structure $\C$ is associated with a Lorentz-Finsler metric $L$ (and, then, with its anisotropically equivalent class). Indeed, if $\C$ is determined by a cone triple $(\Omega,T,F)$, one can construct such an $L$ starting at the map \begin{equation}
\label{eG} G(v):=\Omega(v)^2-F(\pi_2(v))^2,  \qquad \forall v\in \bar A,\end{equation}
where $\pi_2: TM = \hbox{Span} (T) \oplus \hbox{ker} (\Omega ) \rightarrow \hbox{ker} (\Omega )$ is the natural projection.
 $G$ satisfies all the required properties of $L$ except the differentiability on Span$(T)$, the latter because of the lack of differentiability of $F^2$ at $0$ when it is not Riemannian.  Indeed, the indicatrix $G^{-1}(1) \subset A$ is not smooth precisely on $T$, that is, only at the point $T_p$ on each $p$. However, standard techniques of smoothability for convex functions allow one to smoothen $G$ around $T$ obtaining the required $L$ \cite[\S 5.2]{JavSan18}.

\item\label{i_5} The lack of differentiability of $G$ above is analogous to the  well-known lack of differentibility of any product of (non-Riemannian) Finsler manifolds. Indeed, if $(M_1,F_1$) is a Finsler manifold then $ dt^2  \oplus (\pm F_1^2)$ are not smooth as Finsler or Lorentz-Finsler metrics on  $\R \times M_1$ along the direction $\partial_t$. This problem   prevents  the extension to the Lorentz-Finsler case of the trivial procedure to construct a relativistic product spacetime starting at a Riemannian manifold. 
 \item\label{i_6}  Given a Lorentz-Finsler metric, there exists a univocally determined $A$-anisotropic connection which is torsion-free and parallel. Moreover, when we consider a properly Lorentz-Finsler metric, this $A$-anisotropic connection can be extended to an open subset $A^*$ which contains $\bar A\setminus \bf 0$. 
 As the extension away from  $\bar A$ is highly non-unique, we will speak about $\bar A$-anisotropic connections. When $A=TM\setminus\bf 0$, we will just say 
  {\em anisotropic connection}\footnote{\label{foot_anisotropic} Essentially, this is a   connection where, formally,  the  Christoffel symbols  of a chart $(U,\varphi)$  depend also on the direction and, so, they are functions on   $TU\cap A\subset TM\setminus \mathbf{0}$,   which are  positive homogeneous of degree zero.  The name  and a  thorough study of $A$-anisotropic connections  were given  in \cite{Jav1, Jav2}; see also  \cite{ECS1, ECS2}  for a  study of   connections  on fiber bundles  from a more general viewpoint.}.  
\een
\end{rem}
Due to this last item, the definitions of some classes of Finsler spacetimes such as the static ones have included the possibility to have some non-smooth directions  \cite{CaStan16, CaStan18, LPH12}.  However, the  smoothing procedure  mentioned  in part (\ref{i_4}) is also applicable to these cases. This shows that, from the foundations viewpoint, the motivation for  non-smooth metrics is not stronger  for the  Lorentz-Finsler case  than for classical relativistic spacetimes \cite[\S 4.2, \S4.4]{JavSan18}.

  \begin{note}\label{note_improperLF} Nevertheless, there are some physical considerations  (see  \S \ref{s_6.1_modifiedSpecialRelativity})  which lead to examples where $L$ is not  smoothly  extendible to  the cone structure $\C$,  even if: 
  
  (i)  its  cone   $\C$ is smooth (so, the cone geodesics are well defined), and 
  
  (ii) the  $A$-anisotropic   connection can be smoothly extended to $\C$ (so,  the  Finslerian curvature tensors  are well-defined on $\C$). 
  
  Such examples could also be included in our definition of  Lorentz-Finsler metrics and spacetimes, as all 
the relevant geometric properties   remain. However, we will consider for simplicity that $L$ is smooth at $\C$ and, when this property does not hold, we refer to them  as {\em improper} and we will discuss whether (i) and (ii) hold then.   Accordingly  (and consistently with \cite[Def. 3.1]{JavSan18}), an {\em improper}  Lorentz-Minkowski norm $L_0$ satisfies all the properties in Def. \ref{d_LMnorm} but the differentiablility at $L_0^{-1}(0)$. 

 Remarkably,  a large class of spacetimes satisfying both conditions  (i) and (ii)  can be found following \cite{PW11}.  Namely,  they hold for any   two-homogeneous function $L$ defined  on  the  set of causal vectors  $\bar A$  determined  by  a cone structure $\C$ such that:  (a) $L$  is zero on $\C$ and  determines a Lorentz-Finsler metric in the interior $A$ of $\bar A$ and  (b)  there is a power of $L$ which is smooth on $\C$ with non-degenerate Hessian     (notice that, in \cite{PW11},  $L$ is assumed to be defined on the whole $TM$). 

 Anyway,   there are some examples of Finsler spacetimes in  the  literature that do not even  satisfy   our weaker definition of improper Lorentz-Finsler spacetime, such 
 as Randers spacetimes or those introduced by Kostelecky \cite{Kos11,KosRus12}, which are the effective model of some particles with no GR background (see the discussion in \cite[Appendix A, B]{JavSan18}). 
  \end{note}

\subsection{Physical intuitions for Finsler spacetimes}\label{s_phys_intuitions}

Next, our aim is to justify physically our  definition of Finsler spacetime (Def.~\ref{d_finslerspacetime}),  supported by some mathematical properties pointed out  above. 
The first consideration is that Postulate \ref{p2} should be regarded now as an approximate symmetry at each point, in a similar way as the affine structure of Postulate \ref{p1} has been regarded as an approximate symmetry to the structure of a relativistic spacetime\footnote{ Even though we focus on the relativistic case,  (disregarding the Leibnizian case and the other possibilities),  one could also consider a {\em Leibniz-Finsler structure} $(\Omega,h)$ on a manifold $M$, where $h$ would be now a Finsler metric on ${\rm Ker}(\Omega)$ instead of a Riemannian one,  according to  Table  \ref{Fig1}.}.  
This means that, now, one cannot find a set of coordinate charts  such that the relations \eqref{e_post2} occur at each $p$; however, one would expect that we will not be far from this situation (at least in regions of spacetime  free of extremely exotic or violent situations). Consistently, we will not have the sets  $\Sm_p$ of linear bases at  each $T_pM$ playing the role of (linear) IFR at $p$. However, one would expect that the set of observers $\So$ introduced in Def. \ref{d_So} will still make sense and will be ``close'' to the space of observers  for a relativistic spacetime. As the latter  is a hyperboloid (asymptotic to a quadratic cone) at each point $p$, now, $\So_p$ should be a strongly concave hypersurface asymptotic to some cone structure defining a  Lorentz-Minkowski norm at $p$ (see Rem. \ref{d_LMnormrem} (\ref{d_LMnorm_i3})) and, moreover, $\So$ should be identified as the indicatrix $\Sigma$ of a Lorentz-Finsler metric $L$.   

\begin{rem} \label{rem_items}
The previous discussion leads us to a Lorentz-Finsler metric $L$ with indicatrix $\Sigma$ equal to $\So$ which lies exactly under our Def. \ref{d_finslerspacetime}  (including also the improper case explained in Note \ref{note_improperLF}). 
  The way to arrive at this definition from the viewpoint of symmetries can be summarized as follows. 

 (1)  Following \cite{GW}, consider the connected parts of the identity $ISO_1(4)$, $SO_1(4)$,  $ISO(3)$, $SO(3)$ of the Poincar\'e,  Lorentz, Euclidean and orthogonal groups, resp. In Special Relativity, the  homogeneous spaces obtained as the  quotients $ISO_1(4)/SO_1(4)$, $ISO_1(4)/ISO(3)$, $ISO_1(4)/SO(3)$ 
are, respectively, the spacetime, the space of all the (rest) spaces (i.e., the space of all the spacelike hyperplanes,  being the standard rest space $ISO(3)/SO(3)$) and the space of observers $\So$ (being the space of velocities $SO_1(4)/SO(3)$). Here, $\So$ is  metrically identifiable with $\R^4\times H^3_+$. 

(2) In General Relativity, $\So$ is identified with the set $\Sigma^g$ of all the future-directed unit vectors.   $\Sigma^g$ is a  subbundle of $TM$ whose fibers are affine hyperboloids at each tangent space. Such hyperboloids characterize $g$  univocally  so that the information of $g$ is codified in $\So$. 

 (3) For the space of observers $\So$ in the Lorentz-Finsler setting,  $\Sigma^g$  is replaced with a hypersurface $\Sigma$ satisfying  formal properties analogous to $\Sigma^g$ (but dropping its pointwise symmetries)  so that it characterizes a Lorentz-Finsler metric.  
\end{rem}
Next, let us discuss more carefully the physical grounds of Def. \ref{d_finslerm}:
\ben\item The fact that $L$ is defined only on a cone domain $A$ and it is extended continuously to $\C$ comes from the nature of the space of observers. 

Recall that, then, one has   timelike geodesics (Rem. \ref{r_finslerspacetime} (\ref{i_2})) and, thus, freely falling observers. 
At least from a  trivial  mathematical viewpoint, this is enough to determine $L$  and, then, the fundamental tensor $g$ on the cone domain $A$. 

Notice that, given an observer  $v\in \Sigma_p$, the  tensor $g_v$ is then also obtained {\em on the directions}
 of $T_v\Sigma_p$.  In principle, $g_v$ (which can be obtained just from $\Sigma$) could be measured, as it comprises properties of neighboring observers.

\item The smooth extensibility of both $L$ and the fundamental tensor $g$  (as a non-degenerate one)  to the cone structure $\C$ appears as a natural approximation  (in principle, one would expect to remain close to the situation in a relativistic spacetime)  which mathematically  ensures that $\C$ is truly a cone (with $S_0$ in Def. \ref{d_scone} satisfying strong convexity).
Moreover, then $L$ also determines lightlike geodesics which, up to reparametrization, are inherent to the cone structure.  The improper case of Finsler spacetimes  satisfying the properties (i) and (ii) 
 in  Note \ref{note_improperLF}   would also satisfy  all these properties about geodesics and cones. 

 Then, as a consequence,  the behaviour of lightlike geodesics becomes completely analogous to the classical relativistic case.  Indeed, Lorentz-Finsler metrics with the same cone structure are also related by an ``anisotropic conformal factor $\mu$'' (Rem. \ref{r_finslerspacetime}(\ref{i_3})) and the cone structure $\C$ also allows one to mimic the relativistic behaviour of Causality (Rem \ref{r_cone}(\ref{r_cone3})).

\item The physical  considerations in the two previous items are also assumed  in  standard  Relativity. 
 Namely,  observers are always  expected to measure only massive or massless particles, 
that is,  elements with velocities in a causal cone.
In General Relativity, this is  apparent from the EPS formulation, where radar coordinates are systematically used with this aim (see the next  section). 
Certainly,   the metric tensor $g$ is assumed to be defined on all the directions in the relativistic case  
but the underlying reason  is that {\em $g$ is fully determined by its value along the causal  vectors}   (Rem.~\ref{r_minkowski}(\ref{r_minkowski5})).  This is not by any means true in the Lorentz-Finsler case, even if $L$ can be extended to the whole $TM$  (recall Rem. \ref{r_finslerspacetime}(\ref{i_1})).

\item When a spacelike separation in a direction $l$ is going to be measured  by an observer $v$,  it seems natural 
to  consider $g_v(l,l)$;  so, it would be irrelevant whether $L$ is not defined outside the cone.
 
Indeed,  from a purely geometric viewpoint, $T_v\Sigma_p$ would be naturally  regarded as the rest space of the observer $v$ at $p$, and $g_v$ would be the unique metric available there, even though the physical process to measure it might not be obvious.  
 It is worth pointing out  here  Ishikawa's 
claim in \cite{Ishi81} that $g_v(l,l)$ can be measured assuming that the physical light rays are 
those of $g_v$. Indeed,  this author   criticizes  Beem's definition of light 
rays,  who constructed them by using  the lightlike vectors on the cone $\C$.  Anyway, in our opinion, Ishikawa's claim needs further physical support. 


\item\label{item5}  It is worth emphasizing that no issue on smoothability occurs with $\Sigma$, which can be assumed  smooth  (as in Rem. \ref{r_minkowski}~\ref{r_minkowski_item2})  in most interesting cases. Indeed:

(a) The Lorentz-Finsler metric $L$ cannot be $C^2$-extended  to $0$, in agreement with the behavior of norms in both the positive definite case and the Lorentz-Finsler one (Rem. \ref{r_minkowski}(\ref{r_minkowski5})).  However,  no physical nor mathematical reason seems to require the smoothability of $L$  at 0 (compare with the EPS approach in \S \ref{5.2.1} below).

(b) Product metrics $-dt^2+F^2$ or, with more generality, the rough Lorentz-Finsler version of static spacetimes $-\Lambda(x) dt^2 +F^2(x,y)$, with  natural coordinates $(x,y)$ at $TM$,  are never smooth at $\partial_t$  whenever $F$ is  Finsler but  not Riemannian.  Consequently, some authors have included  the possible existence of non-smooth directions as a fundamental ingredient of Lorentz-Finsler metrics  (see for example \cite{CaStan16, CaStan18, LPH12}). Nevertheless, as explained   in 
 Rem. \ref{r_finslerspacetime}, parts (\ref{i_4}) and (\ref{i_5}), general smoothing procedures can be applied. What is more, a natural definition of (smooth) static spacetimes as well as  
  an  explicit procedure to construct locally all of them are available at \cite[\S 4.2]{JavSan18}.
 
 (c) Other issues of non-smoothness appear when modelling some specific physical situations (Very Special Relativity, birefringence) and will be considered in  \S \ref{s_6.1_modifiedSpecialRelativity}. 
\een



\section{Comparison with Ehlers-Pirani-Schild approach}\label{s5}

\subsection{Summary of the approach} \label{s5.1} EPS approach   \cite{EPS}  constructs step by step each geometric structure of physical spacetime (until reaching the metric) by means of physically motivated axioms:

\ben\item\label{EPS_1} Spacetime becomes a {\em differential manifold} $M$ endowed with a {\em cone structure} $\C$. 
Essentially, this is obtained by means of axioms on light propagation which involve messages and echoes between particles. 

Indeed, these axioms  allow one to find  {\em radar coordinates} with respect to {\em  (freely falling, massive)   particles},  the latter represented  by a class of unparametrized curves, which provide the structure of differentiable manifold, see EPS axioms $D_1$---$D_4$. Then, the cone structure $\C$ is obtained by using two axioms, $L_1, L_2$, on the local character of light propagation  around each  event $e$.  Indeed, $L_1$  states 
that given any   particle  P  with some parameter $t$  which passes through  $e$,  it follows that   any event  $p$ ($p\not\in$ P)  can be connected with the particle by exactly two light 
rays\footnote{\label{foor_EPS_DiffStr}Along the events $\tilde e\in$ P, all the light rays from $\tilde e$ would  trivially   cross  P at $\tilde e$; so, the function $g$ below would be trivially extended as $g(\tilde e)=-t(\tilde e)^2$.  However,  the points on $P$ would be excluded in order to define the differentiable structure of the manifold by using radar coordinates (recall the example in footnote 
\ref{foot_EPS_echo} below).},  while $L_2$ distinguishes two connected components for light rays.  Moreover, $L_1$ also states that, if these two rays cross the curve at the events $e_1,e_2$, then  $ g(p):=-t(e_1)t(e_2)$  is required to be  smooth  in a small neighborhood of $e$. 
 EPS claims that, then, $\C$ will come from the {\em conformal structure}  of some  {\em Lorentz metric}  (a particular case of our Def. \ref{d_cone}) and, so, we can speak  about {\em $\C$-timelike}  directions. 

\item\label{EPS_2} Spacetime is endowed with a {\em projective structure} $\mathcal{P}$. This is achieved by means of two axioms, $P_1$, $P_2$,  which model the free fall of  particles.  
 
The first one states only the existence of a  unique 
 particle, represented by means of an (unparametrized) curve, for each event $e$ and $\C$-timelike direction at $e$. The second axiom states that, around each event $e$, one can find coordinates  $\bar x^i$ such that  {\em any} 
  particle through $e$ admits a parametrization $\bar x(\bar u)$ satisfying:

\begin{equation}\label{e_law_of_inertia}
\left. \frac{d^2\bar x^i}{d\bar u^2}\right| _e=0.
\end{equation}
This equality is regarded as an {\em infinitesimal law of inertia} (consistently with Trautman \cite{Trautman}). By using \eqref{e_law_of_inertia}, EPS argues that a {\em projective structure},  which is claimed to be  compatible with some affine  connection $\mathcal{A}$, must appear. As a consequence,    not only the original particles would be recovered as pregeodesics   of  $\mathcal{A}$ but one would also obtain  {\em pregeodesics  at any  direction}, timelike or not.

\item\label{EPS_3} Spacetime is a {\em Weyl space} $(M, \C,\mathcal{A})$, where   $\mathcal{A}$ is an affine connection {\em compatible} with the cone structure $\C$, in the sense that the lightlike $\C$-pregeodesics are also $\mathcal{A}$-pregeodesics. 	This is carried out by means of their axiom $C$, which matches particles and light rays.

Specifically, this axiom assumes that, around each event $e$, any point in the $\C$-chronological future of $e$ lies on a particle through $e$. This 
will imply that the lightlike $\C$-pregeodesics  of the conformal structure  (namely,  the $\C$-cone geodesics, see Remark \ref{r_finslerspacetime}~\ref{i_2})   are also pregeodesics for the projective structure $\mathcal{P}$ in the step (\ref{EPS_2}). Then, EPS claims  that such a compatibility selects a  unique affine connection $\mathcal{A}$ compatible with the projective structure. 

\item\label{EPS_4} Spacetime is endowed with a {\em (time-oriented) Lorentzian metric $g$}, up to an overall  (constant)  scalar factor. This is obtained by means of a {\em Riemannian axiom}, which takes  into acccount that $\mathcal{A}$ has its own parallel transport and  its  curvature tensor; the axiom imposes the compatibility of (one of)  these  two elements with $g$.

Indeed, they state that 
the Riemannian compatibility of $(M, \C,\mathcal{A})$ is equivalent to 
  any of the following conditions: (a) the vectors obtained by  $\mathcal{A}$-parallel transport of  a single one $v$ at  $p\in M$ along two curves with the same endpoint $q$ have the same norm at $q$  (computed with  any of the homothetic scalar products compatible with $\C_q$), or (b) using Jacobi fields to construct arbitrarily close particles, the proper times of 
  two  of  such particles are linearly related at first order, that is, 
the regular ticking of a clock for the first particle  implies the regular ticking for the second one. 
\een
 About these axioms and proofs, EPS admits: ``a fully rigorous formalization  has not yet been achieved''. Next, we will  focus just on the relation of EPS approach with Lorentz-Finsler metrics.  For progress on EPS approach,  see for example 
\cite{Stachel}. 

\subsection{Keys of compatibility with Finslerian spacetimes} \label{5.2} The fact that a Finslerian spacetime  can fulfill  the EPS axioms  was already pointed out by
 Tavakol \& Van den Berg \cite{Tavakol}, 
 who considered the case of Berwald spaces.  Now, our aim is to revisit precisely the compatibility of the four EPS steps with Finslerian elements,  as well as  \cite{Tavakol}. 
 
 \subsubsection{EPS step (\ref{EPS_1})}\label{5.2.1}
Recently,  Lamm\"erzhal and Perlick \cite{LP} have argued against the role of smoothness of the function $g(p)$ at $e$ in the step (\ref{EPS_1}). This differentiability becomes essential, because  the equalities $g(e)=0, g_{,a}(e)=0$  allow EPS to find a metric tensor  $g_{,ab}(e)$ compatible with $\C$. 

\smallskip

\noindent
Indeed,  there are subtle differences at this point in comparison with the introduction  of radar coordinates, which are used to settle  the smooth ($C^3$) manifold structure of the spacetime.
Certainly, EPS were  aware  of the 
existence of non-trivial subtleties,   as one can read at the beginning of their subsection {\em Differential Topology}: ``The reason that we do not take this structure [smooth manifold] for granted is that differentiability  plays a crucial role in our introduction of null cones  (...) and in the infinitesimal version of the law of free fall''.  The following three items must be taken into account in the EPS development: 

(i)  The axioms $D_1$---$D_4$, which allow one to define radar coordinates,   should  apply to particles P, Q which do not intersect. Otherwise, spurious differential issues might appear even in the case of Lorentz-Minkowski spacetime\footnote{\label{foot_EPS_echo} For  example, 
 let P be the $t$-axis and Q$=\{(t,x=t/2,y=0,z=0): t\in\R\}$. A message from Q to P would yield the map $t\mapsto t/2$ if $t\leq 0$ and $t\mapsto 3t/2$ if $t\geq 0$ (see the Example \ref{ejemplo_fundamental} below) which is not smooth at 0, in contradiction with $D_2$  (recall also 
 footnote~\ref{foor_EPS_DiffStr}).}.

(ii)  Axiom $L_1$, however, considers the functions $p\mapsto t(e_1), p\mapsto t(e_2)$ (which would be radar coordinates for some particle P through $e$) defined even on P. 
Moreover, this  axiom  ensures  that the particular combination $g(p) = -t(e_1)t(e_2)$ is $C^2$-differentiable  on P too.

(iii)  In  the discussion   above \cite[Lemma 1]{EPS}, they explain that $t(e_1)=t(e_2)=0$  occurs  if and only if $p=e=e_1=e_2$ (thus, $p\in$ P) and  they focus on this case.
Then, EPS argues first that the differential $g_{,a}(e)$ must be 0 by applying $L_2$ and, using $C^2$ differentiability, they show that the light directions must lie in the quadratic cone of the lightlike vectors of  $ g_{,ab}(e)$. 

\smallskip

\noindent  Recall, however, that there is {\em no physical justification} about why $g$ must be differentiable or $C^2$. Notice that $g$ is constructed from the functions $e\mapsto t(e_1)$ and $e\mapsto t(e_2)$, which are not smooth even in  the Lorentz-Minkowski spacetime  (see Example \ref{ejemplo_fundamental} below).  This assumption on the product $t(e_1)t(e_2)$ yields a posteriori the quadratic character of the cone, forbidding more general cone structures. 

 From a purely mathematical viewpoint, 
the smoothness issue  on the radar coordinates above would be similar to the differentiability  of the radial coordinate $r$ of a normed vector space at $0$:  $r$ is never smooth at 0 and  $r^2$ is smooth if and only if the norm comes from a Euclidean scalar product  (Rem.
\ref{r_minkowski}(\ref{r_minkowski5})). So, such an a priori assumption would be completely {\em unjustified from a mathematical viewpoint}  too,  indeed: 

\bit
\item[(a)] There are norms with an analytic indicatrix (thus, analytic away from 0) which do not come from a scalar product. For example,  on $\R^2$,  when the indicatrix is equal to the   curve in polar coordinates $\rho(\theta)= 1+ \epsilon \sin\theta$   for small $\epsilon>0$ (so that it is strongly convex). 
 \item[(b)] Euclidean scalar products are  very particular cases of analytic 
 norms.
\eit
That is, {\em the apparently mild EPS requirement of differentiability  at 0 becomes even stronger 
 than analyticity for a norm.} 

\begin{exe}\label{ejemplo_fundamental}
Let us see the role of smoothability for the EPS  function $g$ obtained by using a pair of radar coordinates with respect to a particle (according to EPS, one should take two pairs of radar coordinates by choosing two particles). We will work on $M= \R\times\R^3$.
 Let $t: \R\times \R^3\rightarrow \R$ be the natural projection, consider any Minkowski norm $F_0$ on $\R^3$ and take spherical-type coordinates  $(r,\theta, \varphi)$ on 
 $\R^3$ (up to suitable points) with $\theta, \varphi$,  the usual spherical angles and $r\equiv F_0$; then, extend the functions $r,\theta, \varphi$ to $\R\times \R^3$
 in a $t$-independent way.   Let $\C$ be  the natural (constant) cone structure   given by $t(p)=r(p)$ and regard the $t$-axis as a particle P. The corresponding radar coordinates are $t\pm r$ and thus, the EPS  $g$ 
 is $g(p)=-t^2(p)+r^2(p)$. 
 This function is smooth at 0 if and only if $F_0$ comes from a Euclidean  scalar product\footnote{Of course, one could introduce a spurious differential structure on $\R^4$ so that $r^2$ becomes smooth for a non-Euclidean $F_0$, but this would not be natural by any means. }. Anyway, the cone structure is smooth, because it is determined by  the  cone triple $(dt,\partial_t,F_0)$ and, so, it is compatible with a smooth  Lorentz-Finsler metric  $L$  (indeed, a Lorentz-Minkowski norm), see Rem.~\ref{r_finslerspacetime}(\ref{i_4}). As stressed in the item (\ref{item5}) below Rem.~\ref{rem_items}, the fact that $-dt^2+F_0^2$ is not smooth at    $\partial_t$ neither contradicts the existence of a smooth  $L$  nor introduces any issue of smoothability.
\end{exe}

 \subsubsection{EPS step (\ref{EPS_2})}\label{5.2.2} The way how EPS deduces the existence of the projective structure $\mathcal{P}$  from the infinitesimal law of inertia \eqref{e_law_of_inertia} consists in rewritting this last formula in arbitrary coordinates to obtain \cite[formula (7)]{EPS}
 \begin{equation}\label{e_law_of_Finsler_inertia}
 \ddot{x}^a + \Pi^a_{bc} \dot{x}^b \dot{x}^c=\lambda \dot x^a
 \end{equation} where $\lambda$ depends on the parameterization $x^a(u)$  of the curve and $\Pi^a_{bc}$ depend on $x^a$. These functions are called the {\em projective coefficients}, as they would determine a projective structure $\mathcal{P}$ compatible with some affine connection.

 However,  if one allowed the functions  $\Pi^a_{bc}$ to depend on the  direction of the  velocities $\dot x^j$, then   $\Pi^a_{bc}(x^i,\dot x^j)$  could represent the formal Christoffel symbols for a Lorentz-Finsler metric $L$  (indeed, for its $A$-anisotropic connection, see Remark \ref{r_finslerspacetime}, item \ref{i_6}).  Thus, the solutions of \eqref{e_law_of_Finsler_inertia} would be pregeodesics for $L$ {\em which satisfy the law of inertia \eqref{e_law_of_inertia},}  up to the following issue of  $C^2$-differentiability of  the chart coordinates  at  the origin.

 The existence of {\em normal coordinates in $\C$-timelike directions} (which is ensured for   any $A$-anisotropic  connection\footnote{In principle, the normal coordinates can be defined when the anisotropic connection is defined for all the vectors in $TM\setminus 0$, 
but it is always possible to extend the $\bar A$-anisotropic connection to all directions locally (see \cite[Remark 6.3]{JavSan18}, where the Lorentz-Finsler case is considered in detail). These coordinates are obtained using the exponential map in a neighborhood as in \cite[Lemma 6.2]{JavSan18}.}) would be the natural mathematical translation for  the law of inertia. However, 
 the Christoffel symbols 
 of a Lorentz-Finsler metric   
 might not be even continuous
 at the origin  by the trivial reason that these symbols may depend on the direction but they cannot vary along each direction (they are homogeneous of degree 0).  Thus, its exponential map is not guaranteed to be $C^2$ at the origin unless 
 the anisotropic connection is    affine (i.e., it  does not depend on the direction). 
 It is known that, for a positive definite Finsler metric, this happens if and only if the metric is of Berwald type\footnote{This means that its Chern-Rund connection defines an affine connection on the
underlying manifold, see \cite{SLK} for quite a few of characterizations.}   \cite{AZ} (see also
\cite[Ex. 5.3.5]{BCS}) and this can be extended to the Lorentz-Finsler case. 
 Indeed,  this type of metrics provides  the Lorentz-Finsler examples  beyond EPS suggested in the literature,
see  \S \ref{s_5.2.5}.

 Summing up, we emphasize: (a) the coordinates provided by the exponential map  of a Lorentz-Finsler metric at any event $e$  are smooth along the half-lines starting at $e$ and they satisfy \eqref{e_law_of_inertia}, and (b) to exclude anisotropic connections because of their  lack of smoothness at 0  is a subtle mathematical issue and (as in the discussion of the Step \ref{EPS_1} in \S \ref{5.2.1}) this is not justified in EPS neither physically nor mathematically. 
Thus,  the law of inertia should be regarded as compatible with Lorentz-Finsler metrics according to our definition (where the directions outside the causal cone are not taken into account), including even the improper case in Note \ref{note_improperLF}).

 \subsubsection{EPS step (\ref{EPS_3})}\label{s.5.2.3} 
 The compatibility of $(\C, \mathcal{P})$ as a Weyl space with a (unique) affine connection $\mathcal{A}$ obtained by using EPS axiom C   becomes  a subtle question. 
On the one hand, Trautman \cite{TrautmanGoldie} claimed the necessity of a detailed proofs in his review on the reprinted EPS article and, shortly after, this author and V. Matveev \cite{MatTrau14} characterized when a pair  $(\C, \mathcal{P})$ is compatible. 
On the other hand, the notion of Weyl space as the triple $(M,\C,\mathcal{A})$ given by EPS  does not coincide with  the standard one of Weyl geometry\footnote{In modern language, a Weyl geometry on $M$ is a conformal structure $\C$ endowed with a connection on the $\R^+$-principle bundle $P\rightarrow \C$, where the fiber of $P$ at each $\C_p$ is the  class of homothetic Lorentzian scalar products compatible with $\C_p$ (see for example \cite{Folland}); 
such a notion was considered in references on EPS as \cite{FF}.}.  
Some authors questioned whether such an EPS structure permits to define a standard Weyl one as well as EPS development at this step. However, very recently, this question has been  positively answered by Matveev and Scholtz \cite{MatSch20}, vindicating the EPS approach. 


 
We emphasize that  the EPS compatibility axiom C can be stated with no modification in the case that $\C$ is any cone structure and $\mathcal{P}$ is the projective class of pregeodesics of any $\bar A$-anisotropic connection defined on all the $\C$-causal directions (as already commented, $\C$ determines intrinsically cone geodesics extending those in EPS conformal cones, Remark \ref{r_cone}(\ref{r_cone3})).  
 So,  the  possibility to extend previous results to this setting  should be explored. 



 \subsubsection{EPS step (\ref{EPS_4})}\label{s_5.2.4}
In the EPS spirit, the {\em Riemann axiom} would be any (minimum, physically well-motivated)  assumption  making a compatible triple $(\C, \mathcal{P}, \mathcal{A})$ also  compatible with a Lorentzian metric, as the conditions labelled (a) and (b) at  step (\ref{EPS_4}). However, in orden to state now a {\em Finslerian axiom}, 
one should notice that   these conditions  involve $\mathcal{A}$ and, so,  they might depend on the way how the previous step is solved. 

 Anyway, it is worth pointing out some reasons which would support the convenience of such a Finslerian axiom. On the  mathematical side, the results collected in Rem. \ref{r_finslerspacetime} (parts  (\ref{i_3}) and (\ref{i_4})) show a natural consistency: (i) any $\C$ can be associated with a Lorentz-Finsler metric $L$, (ii) any other associated $L'$ is anisotropically related to $L$, and (iii) the lightlike pregeodesics of all the associated Lorentz-Finsler metrics agree with the cone geodesics of $\C$.  
 On the physical side, the standard chronometric approach  is  reduced to the determination of the indicatrix of the observers at each event and this would depend only on the behaviour of clocks  and measurements of proper time\footnote{Compare with EPS claim (\ref{EPSclaim_1}) in \S \ref{s.5.3} below.}.  Notice that, in the Finslerian case, this behaviour would not be restricted  by any condition of quadratic compatibility (but only by a mild overall  concaveness and asymptoticity to~$\C$).

\subsubsection{Finslerian examples strictly compatible with EPS}\label{s_5.2.5}
As we have explained, the requirement of  $C^2$ smoothability at 0 for cones and geodesics is the main gap in the EPS approach. However, Tavakol and van der Berg \cite{Tavakol} showed Finslerian examples which are  even compatible with  this requirement. Next, let us analyze these and other possible examples of  Finsler EPS compatible (FEPS) spacetimes.

A very simple FEPS example would be the following. Consider
an affine space
endowed with any
Lorentz norm $L_0
$ with the same
cone as a
Lorentzian scalar
product $\langle
\cdot,\cdot\rangle
$ ($L_0$ can be obtained by perturbing the indicatrix of $\langle
\cdot,\cdot\rangle$, as explained in Remark \ref{r_finslerspacetime}, item (\ref{i_4})). Then, the cone
and geodesics of
$L_0$ would satisfy
all the EPS axioms, including those of $C^2$ smoothness at 0. Here, the key is that the affine parallel transport preserves both, the indicatrix of $L_0$ and $\langle
\cdot,\cdot\rangle$.

\begin{rem}  Tavakol \& van
der Berg examples also obey this pattern, even though they are  more refined and interesting. Indeed, they are Berwald type spacetimes constructed by using an auxiliary Lorentz metric $g$. The fact that  they are  FEPS examples becomes apparent, because they have the same cone and geodesics as $g$. 
\end{rem}

However, we emphasize that these FEPS examples are {\em not} in contradiction with the EPS conclusions. Indeed, the above examples only show that the physical elements $\mathcal{C}$, $\mathcal{A}$, 
under the EPS restrictions, may be compatible with two different geometric structures, the Lorentz $g$ and  Lorentz-Finsler $L$ metrics. To decide which of them would be physically more appropriate would depend on further physical input. In  absence of such input,  the use of $g$ would be mathematically simpler. Nevertheless, this input might appear from the measurements of proper time, as suggested at the end of \S \ref{s_5.2.4}.   

In order to obtain a true Finslerian contradiction with EPS conclusions, one should construct a Lorentz-Finsler metric $L$ with associated cone  $\C$ and anisotropic connection $\mathcal{A}$ satisfying: 

(i) the EPS $C^2$ requirements, 

(ii) the cone geodesics of $\C$ are pregeodesics of $\mathcal{A}$,  

(iii) $\C$ is invariant under the  $\mathcal{A}$-parallel transport,  and 

(iv) $\mathcal{A}$  
is not compatible with any Lorentz metric. 

\smallskip

However, the following known results on linear algebra and Finsler metrics suggest  the difficulty to find such a contradiction.  Notice that the Finslerian results have been obtained in the positive definite case (the last one after the original EPS paper)  and their suitable extensions to the Lorentz-Finsler case is not always clear:

(a) {\em The square of a norm is $C^2$ at 0 if and only if it comes from an Euclidean scalar product} (\cite{Warner65},   \S \ref{s_phys_intuitions}, item \ref{item5} (b)). As a consequence,  the   $C^2$ requirement (i) implies the Lorentzian character of the cones, \S \ref{5.2.1}.

(b) {\em A linear map between two Lorentzian vector spaces is homothetic if and only if it preserves the lightcones}
\footnote{See for example \cite[\S 2.3]{BEE}.}. As a consequence, if $\C$ is compatible with a Lorentzian metric $g$ (as established in (a)),  the preservation of  $\C$ under   $\mathcal{A}$-transport  in (iii) implies that this transport must be a $g$-homothety; in particular, the Riemann axiom (its version (a) in \S \ref{s5.1}, item \ref{EPS_4}) is satisfied.

(c) {\em The exponential of a Finsler metric is smooth at 0 if and only if it is Berwald}  \cite{AZ}.  As a consequence,  the law of the inertia (with the $C^2$ requirement (i)) would imply that only Berwald-type Lorentz-Finsler metrics could be admitted, \S \ref{5.2.2}).

(d) {\em All Finsler metrics of
Berwald type metric  are affinely equivalent to a Riemann space}, that is, their affine connections
are Levi-Civita for  Riemannian metrics (Szab\'o, \cite{Sz}).

\smallskip

Notice that, in the case that a suitable Lorentz-Finsler version of this last result  existed (taking into account, eventually, the requirement (ii)), this would imply that FEPS is also compatible with a Lorentz metric, that is, the requirement (iv) could not be fulfilled if (i), (ii), (iii) held.

\begin{rem}  \label{r_Fuster0}
 Recently, Fuster et al \cite{FHPV} have shown that there are Berwald-type Finsler spacetimes which are not affinely equivalent to a Lorentz metric. However, they contain non smooth directions;  this must be taken into account for the comparison with Szab\'o's result or the possible contradiction with EPS.   Anyway, they show a minimal violation of smoothness.  Indeed, their examples include  improper Lorentz-Finsler metrics $L$, satisfying both (i) and (ii) in Note \ref{note_improperLF} and,  moreover, they   satisfy that some power  $L^r$ (with $r>1$ and integer) is smooth even at the lightlike directions of their cone, see Remark \ref{r_Fuster1}. 
\end{rem}

\begin{rem} Recently, Hohmann et al \cite{HPV1} have  classified the Berwald
spacetimes which are spatially homogeneous and isotropic. Among them, they have found a genuinely Finslerian class (with cones equal to classical FLWR spacetimes).
As  a proper  Finslerian extension of relativistic  cosmological spacetimes, 
the interest of this FEPS class is remarkable (even if it is not clear that they yield a true contradiction with EPS or not). 
\end{rem}
 
\subsection{Constructive EPS approach vs observer's approach}\label{s.5.3}
In order to compare EPS approach and ours, notice first that EPS distinguishes between a {\em chronometric} approach \`a la Synge \cite{Synge} and their {\em constructive} approach. The former one regards  the concepts of particle
and standard clock as basic, and introduces  the metric $g$ as fundamental. So, it regards as  primitive an easily measurable  physical quantity (proper time) and a single geometric structure (the metric), the latter encoding all the other geometric elements in a simple way. As a consequence of these  advantages, the chronometrical approach is  very economical. However, EPS also pointed out drawbacks such as:
\begin{enumerate}
\item 
 \label{EPSclaim_1} the impossibility to construct the metric from the behavior of the clocks alone, \item 
 the inclusion by hand of the hypothesis that  metric geodesics will correspond with free motion and, then, 
 \item 
 the expectation that the clocks constructed by means of freely falling particles and light rays will agree with the metric clocks. 
\end{enumerate} This motivated  their constructive approach starting at basic elements  (events, particles, light rays) and axioms close to the physical experience. Certainly, EPS aimed to deduce the metric structure from their axioms.  However, the difficulties found in some points (as explained in \S \ref{s.5.2.3}, the step  (\ref{EPS_3}) would have been solved only very recently) as well as the necessity to introduce a Riemannian axiom at the end, makes the procedure somewhat awkward. 

In contrast, our approach is neither chronometric nor constructive; instead, it only appeals to the way how we  measure. As such a procedure is  complex, one starts at the ideal situation when some symmetries among measurements are assumed   (our two postulates).  Under our viewpoint if such symmetries did not hold at all, it would  not be clear even the meaning of the verb ``to measure''. However, in the case that the symmetries can be invoked as an approximation, the meaning of measurements can be recovered.
 Then, the emergence of some geometric structures resembles a sort of experimental Klein's Erlangen program. 

Notice that only hypotheses on  the way of taking coordinates of space and time (inertial reference frames, observers) were assumed.
 It is noteworthy that only some few possibilities emerged for the geometry of spacetime when these symmetries hold in a strict way.  From the standard physical viewpoint  (close to philosophical realism), 
the fact that the   space, time and matter allow us to measure in  some specific way should be interpreted as an evidence about the power  of the emerged geometric structures in order to describe the physical spacetime.

Anyway, it is also worth noticing that our final geometric model of spacetime (a manifold endowed with a Lorentz-Finsler metric defined only on the set $\bar A$ of causal vectors for a cone structure) is compatible with EPS approach. Indeed,  as shown in the previous subsection, EPS excluded the properly Finslerian case  only due to two mathematical subtleties  about  unjustified restrictions of smoothness in radar coordinates (step (\ref{EPS_1}))  and   the  law of inertia   (step (\ref{EPS_2})). 
    As pointed out in our discussion at  \S \ref{s_5.2.4},  in the case that $\C$ (or the Weyl pair $(\C, \mathcal{P})$ in the step (\ref{EPS_4})) were not assumed to be compatible with a Lorentz metric, the Riemannian axiom might be replaced by   a Finslerian one which would involve only the behaviour of clocks. 
  
Finally, we  emphasize that 
EPS approach also gives a strong support to our hypothesis that, in principle, the Lorentz-Finsler metric must be defined only at the causal directions in $\bar A$: no basic element  in the EPS approach (particle, light rays, radar coordinates,  echoes) involves non-causal directions. 

\section{Lorentz symmetry breaking}\label{s6}

The implications of the introduction of Finslerian geometry may be more transparent if we focus on  the  Lorentz symmetry breaking
which occurs
when  Lorenz-Finsler norms are used to extend   Special Relativity  (i.e., when one considers only the second non-linearization  in Table~\ref{tab1}).  We will center around this breaking from our theoretical viewpoint; for a more experimental one, a review  on tests of Lorentz invariance  (which includes Lorentz-Finsler possibilities and discussions on von Ignatowski approach) was updated in 2013 by Liberati \cite{Liberati}. 

\subsection{Modified Special Relativity }\label{s_6.1_modifiedSpecialRelativity}
Assume that the  spacetime has a structure of affine $n$-space Aff and it is endowed with  a  Lorentz-Minkowski norm  $L_0$  rather than a Lorentz scalar product $\langle\cdot,\cdot\rangle_1$.  
Roughly speaking, this is a generalization of Special Relativity where, instead of dropping Postulate~\ref{p1} (as in General Relativity), we are dropping Postulate \ref{p2}. Thus, one has affine reference frames but no IFR's;  however,  one can still assume that any physically relevant vector basis $B$  will be composed of a timelike vector with respect to the cone $\C_0$ associated with $L_0$ and three  non-causal  ones  spanning a spacelike hyperplane $\Pi$ ($\Pi\cap \C_0=\emptyset$). 
\begin{rem}\label{r_6.1} There is a  mathematical analogy between the  transition from $\langle\cdot,\cdot\rangle_1$ to $L_0$ and  the one from Special to General Relativity. The latter   goes from the point-independent $\langle\cdot,\cdot\rangle_1$  to a  Lorentz metric $g_p$ which depends on the point $p$ in an $n$-manifold $M$. In the  former transition 
the vector space $V$ associated with  Aff is  endowed with a 
 Lorentzian metric $g_v$ which depends on the direction of   $v\in \bar A_0$ for some cone  structure  $\C_0$. 
What is more, the independence of $g_v$ with the radial direction ($g_v=g_{\lambda v}$ for $\lambda>0$) makes  relevant only the variation of $v$ on a topological $(n-1)$-spherical cap. 
\end{rem}

\subsubsection{VSR and GVSR  }
The transition from $\langle\cdot,\cdot\rangle_1$ to $L_0$ appears naturally in the so-called {\em Very Special Relativity}  (VSR). This was introduced by Cohen and Glashow \cite{CG} who realized that most physical theories (including   those satisfying the charge-parity symmetry) which are invariant under  certain  proper subgroups of the Poincar\'e group  have  the symmetries of Special Relativity.  Thus, the cases when VSR does not imply Special Relativity appear as a convenient arena to test violations of Lorentz invariance. 
Remarkably, Bogoslovsky \cite{Bogoslovsky} had already studied 
the most general transformations which preserve the massless wave equation
and he found the invariant metric:
\begin{equation}
\label{e_Bog}
L_{\hbox{{\tiny Bog}}}=\langle \cdot,\cdot \rangle_1^{(1-b)}(\beta\otimes\beta)^{b} ,
\end{equation}
where $\beta$ is a $\langle \cdot,\cdot \rangle_1$-lightlike dual vector and $0\leq b<1$, a constant\footnote{$\beta$  would correspond with the direction of propagation of the wave, $b$ with a parameter for a conformal transformation of $\langle \cdot,\cdot \rangle_1$ which preserves the wave equation and $L_{\hbox{{\tiny Bog}}}$ with a Finsler metric invariant by this transformation, see also \cite{Bog18} for further information.}.

\begin{rem}\label{rem_bog} 
 (1)  When $L_{\hbox{{\tiny Bog}}}$ is 
restricted to the future causal cone $\C_0$ of 
$\langle \cdot,\cdot \rangle_1$, then it becomes a Lorentz-Minkowski norm, up to the requirement of 
differentiability at the lightlike vectors,  that is, $L_{\hbox{{\tiny Bog}}}$ is an {\em improper} Lorentz-Minkowski norm according to Note \ref{note_improperLF}. Indeed, $L_{\hbox{{\tiny Bog}}}$ is not smooth at  $\C_0$, 
but it trivially satisfies  the properties (i) and (ii)  of that  note   as, in this case, the $A$-anisotropic  Chern  connection of  $L_{\hbox{{\tiny Bog}}}$  is the affine connection of the Euclidean space.  

 (2) Recall that the restriction of $L_{\hbox{{\tiny Bog}}}$ to  the causal $\C_0$-vectors   is  natural not only  because of  
the physical reasons discussed in the previous sections, but also because the vectors where $L_{\hbox{{\tiny Bog}}}$ vanishes  include the 
 $\langle \cdot,\cdot \rangle_1$-spacelike ones in the kernel of $\beta$, and these vectors do not seem to admit any natural interpretation as directions of light rays.
 
\end{rem}

 As a generalization of VSR for curved spaces, {\em General Very Special Relativity (GVSR)}  drops the invariance of VSR by translations. This was introduced by  Gibbons et al. \cite{GGP07}, who  
  pointed out the Finslerian character of  GVSR.  Relevant examples of Lorentz-Finsler metrics in VSR and GVSR have been recently found, see  \cite{FP16, FPP18} and references therein. 

\begin{rem}\label{r_Fuster1}
A natural generalization of Bogoslovski metric to GVSR is obtained by regarding $\langle \cdot,\cdot \rangle_1$ and $\beta$ as a Lorentz metric and arbitrary 1-form on a manifold $M$. Fuster et al. \cite{FHPV} even consider the generalization obtained by multiplying the latter by a  homogeneous factor type $(c+m \beta^2/\langle \cdot,\cdot \rangle_1)^p$, where $c,m,p\in\R$.
  Among this type of metrics, they found the Berwald spacetimes non-affinely equivalent to a Lorentz one  cited in Remark \ref{r_Fuster0}.
\end{rem}
 
\subsubsection{Smoothability  at  the cone and birefringence}\label{r_Bog}  By starting at our previous study of Bogoslovsky metric, we can go further in the issue of the differentiability of the Lorentz-Finsler metrics at the cone, by comparing our approach with the one introduced by Pfeifer and Wohlfart  (PW) \cite[\S A]{PW11}, which has been modified sometimes \cite{HPV, HPV1}. 

These authors  considered a definition of Lorentz-Finsler 
spacetime and metric  which permits degenerate  directions. 
This  definition  becomes consistent with our  notion  of 
improper Lorentz-Finsler metric in Note~\ref{note_improperLF} and the conditions (i) and (ii) therein. 
Essentially, PW considers, instead of a Lorentz-Finsler metric $L$ as above,   a function $L_r$ which is   $r$-homogeneous for some  $r\geq 
2$, and they relax the non-degeneracy of the fundamental tensor $g$ allowing a set of zero-measure where it degenerates.  Remarkably, 
the smoothness of $L_r$ does not imply the smoothness of the two-homogeneous function $L=L_r^{2/r}$  along the cone $\C$.  Nevertheless,  the $A$-anisotropic connection (which is well defined on a dense set of timelike vectors) can be then extended to the lightlike ones  (see \cite[Th. 2]{PW11}).  
 In this case, $L=L_r^{2/r}$ lies under our definition of improper Lorentz-Finsler metric with a connection extendible to $\C$.

 However, for most choices of  $b$, Bogoslovsky metric 
\eqref{e_Bog} is an example which does not lie under 
 PW definition, in spite of having a regular cone and a connection extendible to it (indeed, both of them the same as Lorentz-Minkowski spacetime).   Nevertheless, they  remain under the definition in the variants \cite{HPV,HPV1} and they are always improper Lorentz-Finsler in the sense of Note~\ref{note_improperLF}, which seems to  provide a suitable geometric framework for these cases. 
Indeed,  let us analyze a generalization of Bogoslovsky metrics from norms to arbitrary manifolds considered  in \cite{FPP18}. 
Let $L_{\hbox{{\tiny Bog}}}=g(\cdot,\cdot)^{(1-b)}(\beta\otimes\beta)^b$, where $g$ is a  (time-oriented)  Lorentzian metric and $\beta$ a  1-form  in a manifold $M$;  notice that, 
whenever $\beta$ remains $g$-causal, the  future cone $\C$ of $g$ agrees with the lightlike vectors for $L_{\hbox{{\tiny Bog}}}$ and this metric is well-defined on all the $g$-causal vectors.   
Let  $r=1/(1-b)$   and  $L^r_{\hbox{{\tiny Bog}}}=g(\cdot,\cdot)(\beta\otimes\beta)^{m}$, with  $m=b/(1-b)$.  Then, 
\begin{multline*} 
g^{L^r_{\hbox{{\tiny Bog}}}}_v(u,w)=\beta(v)^{m}g(u,w)
+m\beta(v)^{m-1}(  g(v,u)\beta(w)+  g(v,w)\beta(u))\\
+  \frac 12  m (m-1)g(v,v)\beta(v)^{m-2}\beta(u)\beta(w).
\end{multline*}
It is not difficult to see that $g^{L^r_{\hbox{{\tiny Bog}}}}$ has the same signature as $g$ when  $\beta(v)>0$ 
(use for example the criterion in \cite[Prop. 4.10]{JavSan18}), but it is trivially equal to zero, when  
$\beta(v)=g(v,v)=0$ and $1/2<b<1$  (observe that in such a case, ${ m } >1$).  
As a consequence, if  $\beta$ is always $g$-timelike,  the generalized Bogoslovsky metric is always a Finsler spacetime  according to PW definition.  When  $\beta$ is $g$-lightlike,   there will be lightlike directions of $L_{\hbox{{\tiny Bog}}}$ which do not satisfy the conditions of PW,  no matter if the  connection is extendible to the (regular, Lorentzian) cone $\C$ or not; however, they will be improper Lorentz-Finsler metrics and satisfy also the definitions in \cite{HPV, HPV1}. 
 An issue beyond the lack of smoothness is birefringence. 
This phenomenon occurs in some crystals and it is described by using two cones, each one with a Lorentz or Lorentz-Finsler metric. It  is related with the dispersion of the light with different wavelengths in the crystal. Some authors have pointed out the  possibility that these dispersions occur also as a constitutive element of the spacetime   \cite{LH, Pfeifer}.  

 One way to describe the lightrays when there is birefringence is by using the product of two Lorentz metrics $L=\sqrt{L_1 L_2}$.  Essentially,  the lightrays are described then by the lightlike geodesics of this product; indeed, when one of the metrics  $L_1$ vanishes and the other does not, then  a metric  anisotropically conformal to $L_1$  is obtained.
  However, some additional subtleties appear. For example,  when the  lightcones $\C_1, \C_2$ of the metrics are one inside the other, say $\C_1< \C_2$, this product   is an 
improper Lorentz-Finsler spacetime
 on the domain  $\bar A_1$ determined by the interior cone $\C_1$ (see \cite[Appendix A.5]{JavSan18}). 
 Notice, however, that the situation would be more complex when the position of the cones is arbitrary. Assuming that the intersection $A_1\cap A_2$ is non-empty  at every $p\in M$, then each $(A_1)_p\cap (A_2)_p$ is convex. However, its boundary may have non-smooth directions and $L$ would become an improper Lorentz-Finsler metric. 

 Under our viewpoint, the existence of 
different light cones may be a worthy possibility (see the discussions around Def. \ref{d_speedoflight}). However, in principle, our mathematical framework would consider separately the cones. Indeed, 
 a possible way  to describe  phenomenons related to the dispersion of light would be to introduce a space  $\bar M= M\times \R^+$ with an extra dimension   representing  the refractive index $\n$.  Then, a Lorentz-Finsler metric $L_{\n}$ would appear for each  $\n$ and the different cone structures $\C_{\n}$ on $TM\times\{{\n}\}$  would project on $TM$.     The birefringent model would
 correspond with an effective description of polarization by using two refractive index, that is, 
 the projection on $M$ of  a limit case on $\bar M$ where only two values of $\n$  would become relevant.  

\subsection{Anisotropic speed of light}\label{s6.3}

In subsection \ref{s3varying_c}, the possibility of a pointwise variation of $c$ was discussed for Lorentz metrics.
As explained  there,   an additional element to the metric structure (such as a pointwise measurement of the fine structure constant $\alpha$) was  germane. 
 Next we will consider some   different possibilities for the measurement of a varying  speed of light (VSL)   proper of the Lorentz-Finsler case. 
 
 
 The underlying reason of the difficulty to measure a VSL in the Lorentzian case relied on the fact that the Levi-Civita parallel transport is a conformal transformation (indeed, an isometry), thus, mapping always affinely lightlike cones into lightlike cones. A first possibility in the Finslerian case is:
 
 \bit
\item[(VSL1)]\label{i_luz1}  Lightlike cones at different  points  may be non-affinely equivalent\footnote{ From a mathematical viewpoint, the property that  lightcones are affinely diffeomorphic is a Berwald-type property. Recall that one of the  characterizations of Berwald manifolds in the class of the Finsler ones is the existence of a torsion free derivative operator such that the parallel translations with respect to it preserve the Finsler norms of tangent vectors \cite[Prop. 6]{SLK}; in particular, the norms at different points are isometric.  }. 
\eit
Clearly, this should be an indicator of the existence of different speeds of light at different points.
Anyway, at the end such a possibility would be possible because a Lorentz-Finsler metric $L$ provides a breaking of Lorentz symmetry at each point. This would turn out in the existence of  anisotropies of the speed of light emitted from a single event $p$ in different directions. So, let us  focus on this possibility, which includes Lorentz-Minkowski norms in affine spaces.

\ben\item[(VSL2)] At an event $p\in M$, a single observer $v\in \Sigma_p$ finds distinct speeds of light at different directions at its {\em rest space} ($T_{v}\Sigma_p$ endowed with $g_v$). 
\end{enumerate}
At least from a purely geometric viewpoint, this could happen as follows.   The cone $\C_p$ will intersect the rest space $T_{v}\Sigma_p$ at some strongly convex  $(n-2)$-hypersurface $S_v$, say, {\em the sky observed by\footnote{Equally, the rest space and the sky  could be regarded as the hyperplane $T^0_{v}\Sigma$ parallel to $T_{v}\Sigma$ through the origin $0\in T_pM$ and  the projection $S^0_{v}$ of $S_{v}$ along the direction $v_p$ into $T^0_{v}\Sigma$, respectively. This is a usual identification in General Relativity, \cite{SaWu77}. } $v$,}  see Fig. \ref{Fig2}. 
\begin{figure}
	\centering
	\begin{tikzpicture}
	\draw
	(-6.1,4.5) -- (-4,0) -- (-1.9,4.5);
	\draw[thick,color=purple] (-5.3,2.45) -- (-6.2,2) -- (-6.2,0.2) -- (-1.7,2.3) -- (-1.7,4) -- (-2.2,3.77);
	\draw[thick, color=blue] (-6,4.5) .. controls (-4,1) .. (-2,4.5);
	\draw[->] (-4,0) -- (-4,1.9); 
	\draw (-4,4) ellipse (1.67cm and 0.5cm);
	\draw[rotate around={26:(-3.88,2)}] (-3.88,2) ellipse (0.943cm and 0.5cm);
	\draw[->,color=red]  (-4,1.9) -- (-3.3,2.6); 
	\draw[->,color=red]  (-4,1.9) -- (-3.3,1.8); 
	\node at (-4,-0.3)(s){$p$};
	\node at (-3.7,2.4)(s){$u_1$};
	\node at (-3.3,2.1)(s){$u_2$};
	\node at (-4.2,1)(s){$v$};
	\node at (-1.6,4.5)(s){$\C_p$};
	\node at (-2.9,4.7)(s){$\Sigma_p$};
	\node at (-2.6,2.3)(s){$S_v$};
	\node at (-1.6,1.8)(s){$T_v\Sigma_p$};
	
	
	\end{tikzpicture}
	\caption{\label{Fig2} The directions $u_1,u_2\in S_v$ may have $g_v(u_1,u_1)\not=g_v(u_2,u_2)$ and, so, the observer $v$ could conclude $c_v(u_1)\not=c_v(u_2)$, i.e., the speed of light depends on the direction.
	}
\end{figure}
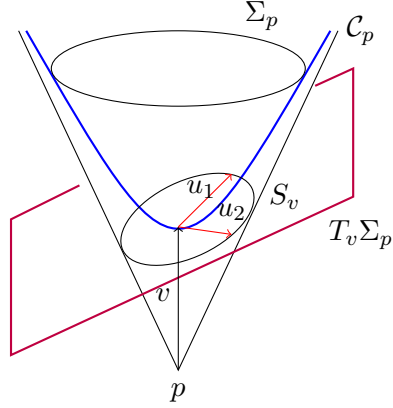
 Then, for $u\in S_v$ the value of $c_v(u):=\sqrt{g_v(u,u)}$ can be regarded as a $u$-dependent speed of light measured by $v$ (namely, the spacelike  length   covered by the light in the direction of $u$ in a unit of time).

\begin{rem}\label{r_vsl2} (1) This $u$-dependent speed of light appears because of the aniso\-tropies of $\Sigma_p$. Thus (in contrast to  (VSL1)) 
it might happen  even for a Lorentz-Finsler metric  compatible with the  cone  
structure 
of Lorentz-Minkowski spacetime (or any other Lorentzian manifold).  

Indeed, at each $p\in M$, the metric $g_{v}$ depends on the space of observers $\Sigma_p$  close to $v$. So, if $\Sigma_p$ were the space of observers for the Lorentz-Minkowski metric $L$,  we could perturb it around some $v\in\Sigma_p$ in order to obtain the space of observers $\Sigma'_p$ of an anisotropically equivalent Lorentz-Finsler metric $L'$ satisfying: 
$$v\in \Sigma_p\cap \Sigma'_p, \qquad \hbox{and} \qquad T_{v}\Sigma_p = T_{v}\Sigma'_p.$$ Then, the skies of $v$ for $L$ and $L'$ are equal but, in general, 
$g_{v}\neq g'_{v}$ and $c_v(u)\neq c'_v(u)$.

(2) It is also worth pointing out that two different observers $v,v'\in \Sigma_p$ will span a single plane $\Pi  \subset T_pM$ which can be regarded as a timelike one for both $g_v$ and $g_{v'}$. The intersections of $\Pi$ with the rest spaces $T_v\Sigma_p$, 
$T_{v'}\Sigma_p$ will give two  lines  $l$ and $l'$ (which are spacelike for $g_v$ and $g_{v'}$, respectively).  Even though $l$ and $l'$ are different they would represent the ``spacelike direction where the other observers lies''. 
However,  the speed of light in the (consistently oriented) directions of $l$ and $l'$ may differ, that is, $c_v(u)\neq  c_{v'}(u')$ for  $u\in l$ and $u'\in l'$ (see  Fig. \ref{Fig3}). 
\end{rem}
\begin{figure}
	\centering
	\begin{tikzpicture}
	\draw
	(-6.1,4.5) -- (-4,0) -- (-1.9,4.5);
	\draw[thick, color=blue] (-6,4.5) .. controls (-4,1) .. (-2,4.5);
	\draw[->] (-4,0) -- (-4.4,2.04); 
	\draw[->] (-4,0) -- (-3.6,2.04); 
	\draw (-6,0.3) -- (-2.78,2.63); 
	\draw (-3.38,1.306) -- (-6,3.2);
	\draw[->,color=red]  (-4.4,2.04) -- (-3.38,1.306); 
	\draw[->,color=red]  (-3.6,2.04) -- (-2.78,2.63); 
	\node at (-3.13,1.4)(s){$u$};
	\node at (-2.58,2.63)(s){$u'$};
	\node at (-6.2,3.2)(s){$l$};
	\node at (-6.2,0.3)(s){$l'$};
	\node at (-3.6,1.04)(s){$v'$};
	\node at (-4.4,1.04)(s){$v$};
	\node at (-1.3,4.6)(s){$\C_p\cap \Pi$};
	\node at (-2.9,4.25)(s){$\Sigma_p\cap \Pi$};
	
	
	\end{tikzpicture}
	\caption{\label{Fig3} In the plane $\pi$ spanned by the observers $v,v'$, the tangent lines to $\Sigma_p$ in $\pi$, $l$ and $l'$ differ. Then if $u\in S_v$ and $u'\in S_{v'}$, possibly, $g_v(u,u)\not=g_{v'}(u',u')$. So the observers $v,v'$ measure different speeds of light in their common plane~$\pi$.
	}
\end{figure}
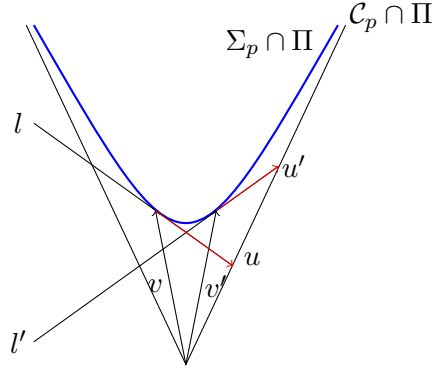

The possibility to measure (VSL2) might be somewhat na\"{\i}ve because: (a) experimental difficulties for the  measurements of the involved geometric elements  $g_v$, $S_v$ (or the relation between $g_v$ and $g_{v'}$), might appear,  and (b)~in the case that $\C$ were compatible with a Lorentzian metric, then one should speak on anisotropies of the space for massive particles (or, eventually, for measurement instruments) rather than for the propagation of light. 

Anyway, there is an {\em  anisotropic propagation of the light} in the case of a breaking of the conformal Lorentz symmetry, namely:
\ben\item[(VSL3)] At an event $p\in M$, the cone $\C_p$ is not compatible with any Lorentz scalar product. 
 \een
In principle, this could be measured by using the  trajectories of  lightrays  even in the case of Lorentz-Minkowski norms  on an affine space (so that the lightrays are straight lines).
Indeed, when $L$ comes from a Lorentz scalar product  $g$, then  $g_{v}$ 
 depends only on $p$ ($g_v\equiv g_p$), $S_{v}$ becomes a sphere  in 
$T_v\Sigma$ centered at 0 of radius $r=1$,  the second fundamental form $\sigma$ (with respect to the inner direction) of $S_v$ can be identified with the restriction of   $g_p/r^2$  to $S_{v}$ and the speed of light is regarded naturally as isotropic. However, in the case of a Lorentz-Minkowski norm $L$, the second fundamental form $\sigma_u$  at some $u\in S_{v}$
  may satisfy, for example,  $\sigma_u > g_{v}/ g_v(u,u)$  (as quadratic forms on $T_uS_v$). Then, the vectors of  $S_{v}$ 
close to $u$ can be regarded as ``shorter'' than those in the Lorentz metric case, that is:
\begin{quote}
{\em the speed of light measured by $v\in \Sigma_p$ at the direction $u\in S_v:=T_v\Sigma_p \cap \C_p$ is bigger (resp. smaller) than the speed of light in the directions close to $u$ when  $\sigma_u > g_{v}/g_v(u,u)$ (resp. $\sigma_u < g_{v}/g_v(u,u)$).}
\end{quote}
More precisely, if $\lambda$ ($>0$) is the eigenvalue of $\sigma_u$ in the direction $w\in T_uS_{v}$ then $1/\lambda$ 
would  rate the increasing of the speed  along the direction $w$. 

 We emphasize that the previous procedure would allow the observer $v$ to realize that an anisotropy holds either in $\Sigma_p$ or in $\C_v$. The fact that $g_v$ only depends  on the behavior of $\Sigma_p$ around $v$ prevents to disregard the first case. However:
\begin{quote}
{\em $\C_v$ is compatible with a Lorentz scalar product  if and only if $S_v$ is an ellipsoid, } 
\end{quote}
and the latter property can be checked in  purely affine terms on $T_v\Sigma$ (namely, it holds when it vanishes the cubic form $C(X,Y,Z)=\nabla_X \sigma^{\xi}(Y,Z)$\footnote{Observe that the cubic form coincides with the Matsumoto tensor of the pseudo-Minkowski norm having the affine hypersurface as indicatrix up to multiplication by a function (see for example \cite{MoHu10} or  \cite{JV}). 
The Matsumoto tensor is zero when the pseudo-Finsler metric comes from a scalar product. } 
 constructed from the second fundamental form $\sigma^{\xi}$ and the induced connection $\nabla$,  both  for the Blaschke normal $\xi$, see \cite[Theorem II.4.5]{NoSa94}.  

\begin{rem}  The property $\nabla \sigma^{\xi}\not\equiv 0$  implies the intrinsic anisotropy of the speed of light, but it does not assign an ``absolute'' speed of light $c_{v}(u)$ (which would depend on the Lorentz metric $L$ as in (VSL2)). However, one has the possibility to measure variations on the speed of light around each  $u$. The qualitative behavior of such variations rely on the cone structure instead of the metric  (compare with Remark \ref{r_vsl2}). 
\end{rem}

\subsection{Matter as anisotropy and  Quantum Physics}\label{s6.4} Clearly, a  Lorentz-Minkowski norm or properly Lorentz-Finsler metric  would appear if some type of anisotropy were detected in physical spacetime (see for example \cite{PerlickFermat} and references therein).   However, we emphasize: 

\begin{quote}
{\em The existence of an anisotropy does not mean necessarily a ``pre-existing spacelike anisotropy of empty space''. Indeed, 
the existence of matter induces anisotropies in causal directions, and this might be reflected in the indicatrix of $L$.} 
\end{quote}
  This possibility is stressed in our formalism, as $L$ is defined only on causal directions. 
Even though this idea is quite  speculative, let us explain it briefly.

Consider first that an  event $p\in M$ is crossed by  a particle  $\gamma$, $\gamma(0)=p$, with mass $m>0$. In this case, $\gamma'(0)$ selects a privileged direction at $p$, and this would introduce  
an anisotropy in the space of observers  $\So_p$ (with respect to a background Lorentz metric). This 	perturbation might be made quantitative in some ways; for example, by introducing a perturbation in the curvature of $\So_p$ around $p$ proportional to $m$. 
In the case of having a stress energy tensor $T$ in an initial background Lorentzian metric $g$, algebraic properties of $T$ (as the energy density or pressure for perfect fluids) might induce the perturbation of $\So_p$.

These perturbations, even if tiny, might have interest at Planck scale. 
Indeed, it is commonplace to assume that nonlinear modifications of linear 
Schr\"odinger equation might lead to an effective collapse which resolves the measurement problem (see for example \cite[\S 7]{FinsterKleiner}). So, the nonlinear framework of Finsler spacetimes opens possibilities  in this direction  which are worth to be studied further.

\begin{rem}
 Recent examples of Finslerian spacetimes, as 
the model of  relativistic kinetic gases in \cite{HPV0}, can be understood also under the above  viewpoint. Their authors explain that an ensemble of a large number of $P$ individual interacting and gravitating point particles can be described at three levels: 

(1) individual particles, 

(2) description as a kinetic gas, by using a 1-particle distribution function (1PDF), which retains information about velocities, and 

(3) description as a fluid, where velocities at each point are also averaged.  

\smallskip

\noindent That reference develops the second viewpoint, where a Lorentz-Finsler model emerges. However, one should take into account that, certainly, the individual particle description is the extreme idealization of the gas, as these particles should be quantum objects. So, the Lorentz-Finsler metric might be directly the most natural description as a semi-classical limit.   
 \end{rem}



\section{Conclusions}
Along this article, we have obtained goals in the following three directions:

\smallskip 

\noindent
(1) A revision  of the foundations of the theories of non-quantum spacetime from the viewpoint of how space and time are being measured, carried out in three parts.

1a. In the first one (doubly linearized models, \S \ref{s2}) the previous approaches in this direction \cite{Ignato, BLS} have been sharpened and simplified, and the four compatible models of spacetime have been concisely described. In particular, we have introduced the hypothesis of {\em apparent temporality}. This hypothesis is enough to obtain the models with no additional hypotheses on, for example, group actions, Theorem \ref{t_k}. Moreover, it will yield time-orientability in three of the models (the temporal ones) and it will underlie our definition of Finsler spacetime, where the Lorentz-Finsler metric is defined only on the causal vectors of a single cone structure.
The other two parts consider their natural non-linear generalizations. 

1b. The  first non-linearization  \S \ref{s3} is carried out in the spirit of the generalization from Special to General Relativity. In a natural way, the previous four models lead to a signature-changing metric, with Leibnizian structures (and their dual) in the  degenerate part and to pointwise variations $c(p)$ of the speed of light which are briefly discussed. It is worth pointing out that, consistently with the discussion at the end of \S \ref{s2}, here $c(p)$ appears as the supremum of velocities between observers at each event $p$; however, it becomes identifiable with the speed of propagation of the light because it propagates in vacuum (and $c(p)$ is the unique common speed different to $0$ measurable by all the observers at $p$). 

1c. Focusing in the Relativistic case, the second non-linearization \S \ref{s4} is obtained just by removing the relativistic quadratic restriction (intrinsic to Lorentzian metrics) on the space of observers. This leads directly to our definition of Finsler spacetime. Its mathematical background and subtleties (including issues on differentiability specific to the Fisnler case which will be relevant later) are also introduced concisely. 

\smallskip 

\noindent
(2) A critical revision of EPS  approach \S \ref{s5}  with a triple aim.

2a. The first aim was to examine which EPS assumptions forbid non-relativistic Lorentz-Finsler metrics to emerge, taking into account previous studies \cite{Tavakol, LP}. We have found that these assumptions appear neatly at two steps (\S \ref {5.2.1}, \ref {5.2.2}) and they have the same origin: they impose certain conditions of $C^2$-differentiability at 0 (in each  tangent space $T_pM$) of some geometric quantities which, by its very nature, forbids any anisotropy and, mathematically,  leads to the quadratic restriction on the metric (the latter, essentially, by an elementary computation 
in \cite[Proposition 4.1]{Warner65}). Intuitively, this condition can be understood as follows: if one has any element in a vector space depending only on the direction (as the fundamental tensor of a non-Riemannian Finslerian metric or the Christoffel symbols of a non-affine anisotropic connection) then this element cannot be even continuous at 0, as this vector can be regarded as the limit of vectors coming from different directions. Of course, such a condition would not be reasonable from a 
mathematical viewpoint (it would exclude as non-smooth even all the analytic 
Finsler metrics) but also from a physical one. Indeed, it would be even  preferable to assume directly the isotropy in different directions as a 
physical assumption,  as such an isotropy might be natural  in some cases. In contrast, the assumption on $C^2$-differentiability at 0 a priori may be 
misleading and it interferes with the assumption on radar coordinates (which is regarded as involved by many authors, see for example recent \cite[footnote 7]{MatSch20}).  For the sake of completeness, we have also studied the Finslerian examples  which are compatible with the EPS axioms (including  $C^2$ differentiability at 0, as in \cite{Tavakol}) and discussed at what extent they contradict EPS conclusions, \S \ref{s_5.2.5}. 

2b. The second aim was to compare EPS, as well as the standard chronometric approach, with ours. As an important difference between the philosophies of the previous approaches and ours, our postulates  do not involve {\em the physical objects which will be measured} but {\em the way how we can measure} physical objects. Indeed, the possibility to make meaningful measurements of the physical spacetime relies on the existence of some mild symmetries among the observers, so that different measurements (carried out at different events and by different observers at each event) can be compared. As stressed here, such symmetries become then apparent in the observers space $\So$ and, then, allow one to determine some geometries for the physical spacetime. The fact that the exact symmetries of $\So$ in the initial linearized model may be only approximate, leads to General Relativity,  modified Special Relativity and the general model of Finsler spacetimes. 

2c. As an extra bonus of the previous two aims,  EPS approach can be also used to obtain Lorentz-Finsler metrics for the geometry of spacetime.
Indeed, removing the criticized hypotheses of $C^2$ smoothability,  any Lorent-Finsler metric $L$ will be compatible with the two first steps of EPS. The other two steps should justify the uniqueness of  $L$ up to an overall factor. 
These steps would be involved mathematically (indeed, the third one would have been justified for the original EPS approach one only recently \cite{MatSch20}). However, 
as suggested in \S \ref{s_5.2.4}, only the behavior of clocks would be enough to construct $\So$ and, then, to characterize $L$. Even though this behavior becomes natural in the chronometric approach rather than in EPS, the main objection of these authors 
to  chronometrics (part (\ref{EPSclaim_1}) in \S \ref{s.5.3}) would be solved.
It is also worth emphasizing that, in this way, our procedure becomes simple and rigorous at all the stages.

(3) A summary of some issues related to Lorentz symmetry breaking discussed from the introduced viewpoint. This includes:

3a. Very Special Relativity and Pfeifer \& Wohlfart (PW) definition of Finsler spacetimes $\S \ref{s_6.1_modifiedSpecialRelativity}$. They are particular cases of  Finsler spacetimes with non-smooth lightlike directions (and, so, they do not satisfy properly our definition of Lorentz-Finsler metric. However, they are endowed with a regular cone  structure $\C$ and an isotropic connection extendible to $\C$ and, so, most of their relevant properties hold (see Note \ref{note_improperLF}). The case of Bogoslovsky metric and its generalization to arbitrary manifolds is studied specifically. Moreover, the way to fit the phenomenon of birefringence in our setting is also discussed.

3b. Three ways to detect the possibility that the speed of light varied with the direction \S \ref{s6.3}. The first one would be a pointwise variation which would go beyond the one  discussed in General Relativity, which relies on the possibility that a cone structure 
has cones at different points non-affinely isomorphic. The other two ways focus on the Lorentz symmetry breaking at each point $p\in M$.  The first one is a geometric analysis which would detect the anisotropies of the Lorentz-Finsler metric $L$ (and, then, of the measured speed of light) in different situations, namely: when a single observer looks at different spacelike directions (Fig. \ref{Fig2}) and when two observers at $p$ compare their spacelike  measurements (Fig. \ref{Fig3}). Because of these anisotropies of $L$, the measured speeds of the light might be different even for a cone structure compatible with a quadratic (relativistic) cone. Thus, the  other procedure focuses on the specific properties of the cone and would detect its lack of quadraticity.

3c. A justification of Lorentz-Finsler anisotropy. Typically, Finslerian anisotropy is considered as a spacelike anisotropy. Notice, however, that our Lorentz-Finsler metrics are not even defined on spacelike directions. As extensively argued along the article, Lorentz-Finsler anisotropies appear on the space of observers. So, it is natural to think that they might be associated with the distribution of mass and energy. These might be anisotropic even if one thought that a ``background isotropic vacuum'' existed. In this vein,  a possible link with Quantum Mechanics is suggested and further developments on this issue might be worthy.

\smallskip

\noindent Summing up, this paper tries to provide physical grounds and precise mathematical formulations for the development of Lorentz-Finsler geometry and its relativistic applications. It is worth emphasizing that the applications, however, go beyond the relativistic setting. For example, an extra bonus has its roots 
in analogue gravity
\cite{BLV}. Indeed, the classical non-relativistic problem of Zermelo navigation
is better understood by using Lorentz-Finsler
metrics and the corresponding Fermat principle \cite{CJS14, JavSan18}. Then, on the one hand,
the classical Finslerian/Zermelo viewpoint has applications to spacetimes
\cite{CJS11, JS17proc} and, on the other, the Lorentz-Finsler viewpoint has applications for
issues such as the propagation of fire spreading, quantum navigation and
classical Finsler Geometry \cite{Markv16,Gib17,JS17}. 
So, Lorentz-Finsler geometry and its applications appears as a fascinating area to be developed further.

\section*{Acknowledgments}
 The authors warmly acknowledge  discussions on this topic with Prof.~ 
Vol\-ker Perlick (Univ. of Bremen), Christian Pfeifer (Univ. of Tartu)  and Nicoletta Voicu (Univ. of Brasov). 

   MAJ  was partially supported by MICINN/FEDER project with reference PGC2018-097046-B-I00  and Fundaci\'on S\'eneca project with reference 19901/GERM/15, Spain, and  MS by 
Spanish  MINECO/FEDER project reference MTM2016-78807-C2-1-P. 
This work is included in the framework of the Programme  of Excellence Groups of the Region of Murcia, Spain funded by Fundaci\'on S\'eneca, Science and Technology Agency of the Region of Murcia.



\end{document}